\documentclass[12pt, a4paper, figuresright]{article}
\usepackage{amssymb, amsthm, amsfonts, mathrsfs}
\usepackage[T2A]{fontenc}
\usepackage[english]{babel}

\usepackage{xcolor}
\usepackage{commath}
\usepackage{thmtools}
\usepackage{thm-restate}
\usepackage{hyperref}
\usepackage{varioref}
\usepackage{cleveref}
\usepackage{natbib}

\usepackage{algorithm}
\usepackage{algpseudocode}

\theoremstyle{plain}
\newtheorem{theorem}{Theorem}%  meant for continuous numbers
\newtheorem{lemma}{Lemma}% 
\theoremstyle{definition}%
\newtheorem{definition}{Definition}%
\newtheorem{problem}{Problem}
\newtheorem{example}{Example}%

\theoremstyle{remark}%
\newtheorem{remark}{Remark}%

\DeclareMathOperator{\RR}{\mathbb{R}}

\DeclareMathOperator{\ZZ}{\mathbb{Z}}

\DeclareMathOperator{\BC}{\mathcal{B}}

\DeclareMathOperator{\PC}{\mathcal{P}}

\DeclareMathOperator{\JC}{\mathcal{J}}
\DeclareMathOperator{\IC}{\mathcal{I}}

\DeclareMathOperator{\NotBC}{\overline{\BC}}

\DeclareMathOperator{\rank}{rank}
\DeclareMathOperator{\cone}{cone.\!hull}
\DeclareMathOperator{\conv}{conv\!.\!hull}

\DeclareMathOperator{\linh}{span}
\DeclareMathOperator{\inth}{span_{\ZZ}}

\DeclareMathOperator{\inputsize}{input\:size}
\DeclareMathOperator{\poly}{poly}

\DeclareMathOperator{\diag}{diag}

\DeclareMathOperator{\disc}{disc}
\DeclareMathOperator{\herdisc}{herdisc}
\DeclareMathOperator{\lindisc}{lindisc}
\DeclareMathOperator{\herlindisc}{herlindisc}

\DeclareMathOperator{\detlb}{detlb}

\DeclareMathOperator{\BUnit}{\mathbf 1}
\DeclareMathOperator{\BZero}{\mathbf 0}

\newcommand*{\intint}[2][1]{\left\{#1,\, \dots,\, #2\right\}}

\newcommand\restr[2]{{% we make the whole thing an ordinary symbol
  \left.\kern-\nulldelimiterspace % automatically resize the bar with \right
  #1 % the function
  \vphantom{\big|} % pretend it's a little taller at normal size
  \right|_{#2} % this is the delimiter
  }}

\DeclareMathOperator{\DiagFrob}{F_{\text{\!diag}}}
\DeclareMathOperator{\SlackFrob}{F_{\text{\!slack}}}

\DeclareMathOperator{\DiagFrobCond}{DiagCondition}
\DeclareMathOperator{\SlackFrobCond}{SlackCondition}
\DeclareMathOperator{\ILPSF}{ILP-SF}
\DeclareMathOperator{\ILPCF}{ILP-CF}
\DeclareMathOperator{\SysSF}{Standard-System}
\DeclareMathOperator{\SysCF}{Canonical-System}
\DeclareMathOperator{\ModularILPSF}{Modular-ILP-SF}

\title{Diagonal Frobenius Number via Gomory's Relaxation and Discrepancy}

\author{Dmitry Gribanov \and Dmitry Malyshev \and Panos Pardalos}
\date{\today}

\begin{document}

\maketitle

\begin{abstract}
    For a matrix $A \in \ZZ^{k \times n}$ of rank $k$, the \emph{diagonal Frobenius number $\DiagFrob(A)$} is defined as the minimum $t \in \ZZ_{\geq 1}$, such that, for any $b \in \inth(A)$, the condition
\begin{equation*}
    \exists x \in \RR_{\geq 0}^n,\, x \geq t \cdot \BUnit \colon \quad b = A x
\end{equation*}
implies that
\begin{equation*}
    \exists z \in \ZZ_{\geq 0}^n \colon\quad b = A z.
\end{equation*}

In this work, we show that
\begin{equation*}
    \DiagFrob(A) = \Delta + O(\log k),
\end{equation*}
where $\Delta$ denotes the maximum absolute value of $k \times k$ sub-determinants of $A$. 

From the computational complexity perspective, we show that the integer vector $z$ can be found by a polynomial-time algorithm for some weaker values of $t$ in the described condition. For example, we can choose $t = O( \Delta \cdot \log k)$ or $t = \Delta + O(\sqrt{k} \cdot \log k)$. Additionally, in the assumption that a $2^k$-time preprocessing is allowed or a base $\BC$ with $\abs{\det A_{\BC}} = \Delta$ is given, we can choose $t = \Delta + O(\log k)$.

Finally, we define a more general notion of the \emph{diagonal Frobenius number for slacks $\SlackFrob(A)$}, which is a generalization of $\DiagFrob(A)$ for canonical-form systems, like $A x \leq b$. All the proofs are mainly done with respect to $\SlackFrob(A)$. The proof technique uses some properties of the Gomory's corner polyhedron relaxation and tools from discrepancy theory.
\end{abstract}
%\author{Gribanov D.\inst{1,2} \and
%Malyshev D.\inst{2}}
%
%\authorrunning{Gribanov D. \& Malyshev D.}
%\institute{Laboratory of Discrete and Combinatorial Optimization, Moscow Institute of Physics and Technology, Institutsky lane 9, Dolgoprudny, Moscow region, 141700, Russian Federation
%\email{dimitry.gribanov@gmail.com}\\
%\and
%Laboratory of Algorithms and Technologies for Network %Analysis, HSE University, 136 Rodionova Ulitsa, Nizhny %Novgorod, 603093, Russian Federation\\
%\email{dsmalyshev@rambler.ru}}
%
%\maketitle% typeset the header of the contribution
%
%\begin{abstract}

%\keywords{TBD}
%\end{abstract}

\section{Introduction}\label{sec:intro}

We consider the following \emph{integer linear feasibility problem in the standard form}:
\begin{problem}\label{prbl:standard_feasibility}
	Let $A \in \ZZ^{k\times n}$, $\rank(A) = k$, $b \in \ZZ^{k}$. 
	Assume that $k \times k$ sub-determinants of $A$ are co-prime, or, equivalently, $\inth(A) = \ZZ^k$, we will clarify this assumption later, see Remark \ref{rm:GCD_assumption}.
	\emph{The integer linear feasibility problem in the standard form of co-dimension $k$} can be formulated as the problem to find an integer feasible solution $z \in \ZZ_{\geq 0}^n$ of the following system:
	\begin{equation}
		\begin{cases}
			A x = b\\
			x \in \RR_{\geq 0}^n.
		\end{cases}\label{eq:Sys-SF}\tag{\(\SysSF\)}
	\end{equation}
%	For simplicity, we assume that $\dim(\PC) = n-k$ for the corresponding polyhedra $\PC = \{x \in \RR_{\geq 0}^n \colon A x = b\}$.
\end{problem}

Following to \citet{DiagonalFrobenius,ILPInTotalRegime}, the \emph{diagonal Frobenius number} is defined  in the following way.
\begin{definition}\label{def:diag_frob}
    Corresponding to the system \eqref{eq:Sys-SF}, the \emph{diagonal Frobenius number $\DiagFrob(A)$} is defined as the minimum $t \in \ZZ_{\geq 0}$ such that, for any $b \in \ZZ^k$, the condition 
    \begin{equation}\label{eq:DiagFrobCondition}\tag{\(\DiagFrobCond(t)\)}
            \exists x \in \RR_{\geq 0}^n,\, x \geq t \cdot \BUnit\colon \quad b = A x
    	% b \in \left\{A x \colon x \geq t \cdot \BUnit \right\} \cap \inth(A)
    \end{equation}
    implies the existence of an integer feasible solution $z \in \ZZ_{\geq 0}^n$ of \eqref{eq:Sys-SF}, that is
    \begin{equation*}
    	\exists z \in \ZZ_{\geq 0}^n:\quad b = A z,
    \end{equation*}
    or in other words,
    \begin{equation*}
    	b \in \cone_{\ZZ}(A).
    \end{equation*}
\end{definition}

In our work, we study the computational complexity of Problem~\ref{prbl:standard_feasibility} and the value of the diagonal Frobenius number $\DiagFrob(A)$ with respect to $k$, $n$, and the absolute values of sub-determinants of $A$. Values of sub-determinants are controlled, using the following notation.
\begin{definition}
	For a matrix $A \in \ZZ^{k \times n}$ and $j \in \intint{k}$, by $$
	\Delta_j(A) = \max\left\{\abs{\det (A_{\IC \JC})} \colon \IC \subseteq \intint k,\, \JC \subseteq \intint n,\, \abs{\IC} = \abs{\JC} = j\right\},
	$$ we denote the maximum absolute value of determinants of all the $j \times j$ sub-matrices of $A$. 
	By $\Delta_{\gcd}(A,j)$, we denote the greatest common divisor of determinants of all the $j \times j$ sub-matrices of $A$.
	% By $\Delta_{\gcd}(A,j)$ and $\Delta_{\lcm}(A,j)$, we denote the greatest common divisor and the least common multiplier of nonzero determinants of all the $j \times j$ sub-matrices of $A$, respectively. 
	Additionally, let $\Delta(A) = \Delta_{\rank(A)}(A)$ and $\Delta_{\gcd}(A) = \Delta_{\gcd}(A,\rank(A))$. A matrix $A$ with $\Delta(A) \leq \Delta$, for some $\Delta > 0$, is called \emph{$\Delta$-modular}. Note that $\Delta_1(A) = \|A\|_{\max}$. 
	% An $\Delta$-modular matrix $A$ with totally non-zero $\rank(A) \times \rank(A)$ sub-determinants is called the \emph{strongly $\Delta$-modular}. 
\end{definition}

By \citet{DiagonalFrobenius}, we have
\begin{equation}\label{eq:AH2010_FUB}
    \DiagFrob(A) \leq \frac{n-k}{2} \sqrt{n \cdot \det(A A^\top)}.
\end{equation}
A significant improvement of \eqref{eq:AH2010_FUB} was recently provided by \citet{ILPInTotalRegime}:
\begin{equation}\label{eq:TotalRegime_FUB}
    \DiagFrob(A) \leq (n - k) \cdot \left( \max\limits_{1 \leq i \leq n} \norm{A_{* i}}_2 \right)^k.
\end{equation}
Additionally, it was shown by \cite{ILPInTotalRegime} that
\begin{equation*}\label{eq:TotalRegime_FLB}
    \DiagFrob(A) > \frac{1}{20 k} \cdot \left( \max\limits_{1 \leq i \leq n} \norm{A_{* i}}_2 \right)^k.
\end{equation*}

The work by \citet{ForallPseudopoly} provides an upper bound, parameterized by $\Delta_1(A)$, which is independent on $n$:
\begin{equation}\label{eq:Forall_FUB}
    \DiagFrob(A) \leq k \cdot \left(2 k \cdot \Delta_1(A) + 1\right)^k.
\end{equation}

Note that, with respect to the upper bound \eqref{eq:TotalRegime_FUB}, satisfying the condition \eqref{eq:DiagFrobCondition} on the diagonal Frobenius number, implies the existence of a polynomial-time algorithm to find an integer feasible solution of \eqref{eq:Sys-SF}.

Our main contribution is a new bound on the diagonal Frobenius number. It has a number of advantages: 
\begin{itemize}
    \item it depends on a weaker parameter $\Delta(A)$,
    \item it is independent of $n$,
    \item it improves upon all the cited estimates after the application of the Hadamard's inequality.
    % \item and can also be improved by knowing a good feasible base of \eqref{eq:Sys-SF}. 
\end{itemize}
It is stated in the following Theorem. Everywhere in the current Subsection, we use the shorthand notation $\Delta := \Delta(A)$.
\begin{restatable}{theorem}{DiagFrobMainTh}\label{th:DiagFrob}
    Denote
    \begin{gather*}
        t_1 = \Delta + C \cdot \log k,\\
        t_2 = \Delta + C \cdot \sqrt{\log k \cdot \log (n-k)},\\
        t = \min\{t_1,t_2\},
    \end{gather*}
    for a sufficiently large absolute constant $C$, whose exact value is not crucial for our purposes. Then,
    \begin{equation*}
        \DiagFrob(A) \leq t.
    \end{equation*}
    Additionally, assuming that a base $\BC$ with $\abs{\det A_{\BC}} = \Delta$ is known, the condition \eqref{eq:DiagFrobCondition}:
    \begin{equation*}
        \exists x \in \RR_{\geq 0}^n,\, x \geq t \cdot \BUnit\colon \quad b = A x
    \end{equation*}
    implies that there exists an integer feasible solution of the system \eqref{eq:Sys-SF}, which can be found by a polynomial-time algorithm.
\end{restatable}

Since it is an NP-hard problem to find a base $\BC$ of $A$ with $\abs{\det A_{\BC}} = \Delta(A)$, the bound of \Cref{th:DiagFrob} does not imply a polynomial-time algorithm to construct an integer feasible solution of \eqref{eq:Sys-SF}. However, there exist weaker upper bounds for $\DiagFrob(A)$, which admit such polynomial-time algorithms. They are presented in the following Theorem.
\begin{restatable}{theorem}{PolynomialDiagFrobMainTh}\label{th:PolyDiagFrob}
    Denote
    \begin{gather*}
        t_1 = \Delta + C_1 \cdot \begin{cases}
            \sqrt{k},\quad\text{for $k \leq n-k$},\\
            \sqrt{k} \cdot \log\bigl(\frac{2k}{n-k}\bigr),\quad\text{for $k \geq n-k$},
        \end{cases}\\
        t_2 = \Delta + C_2 \cdot \Delta \cdot \log k,\\
        t_3 = \Delta + C_2 \cdot \Delta \cdot \sqrt{\log k \cdot \log (n-k)},\\
        t = \min\{t_1,t_2,t_3\},
    \end{gather*}
    for sufficiently large absolute constants $C_1,C_2$, whose exact values are not crucial for our purposes.
    Then, the condition \eqref{eq:DiagFrobCondition}:
    \begin{equation*}
        \exists x \in \RR_{\geq 0}^n,\, x \geq t \cdot \BUnit\colon \quad b = A x
    \end{equation*}
    implies that there exists an integer feasible solution of the system \eqref{eq:Sys-SF}, which can be found by a polynomial-time algorithm.
\end{restatable}

Another interesting case is when the co-dimension parameter $k$ is a constant or a slowly growing function, depending on the input size. If this situation occurs, there exists an upper bound on $\DiagFrob(A)$, which admits a $2^{k} \cdot \poly(\inputsize)$-time algorithm to find an integer feasible solution of \eqref{eq:Sys-SF}. The constant $C$ in the corresponding upper bound on $\DiagFrob(A)$ is $e^2$ times larger than the constant, we found in \Cref{th:DiagFrob}.
\begin{restatable}{theorem}{ExpDiagFrobMainTh}\label{th:ExpDiagFrob}
    Denote 
    \begin{gather*}
        t_1 = \Delta + C \cdot \log k,\\
        t_2 = \Delta + C \cdot \sqrt{\log k \cdot \log (n-k)},\\
        t = \min\{t_1,t_2\},
    \end{gather*}
    for a sufficiently large absolute constant $C$, whose exact value is not crucial for our purpose $($the constant $C$ in this Theorem is $e^2$ times larger than the constant in the Frobenius number $\DiagFrob(A)$, we found in \Cref{th:DiagFrob}$)$.
    % \begin{equation*}
    %     t = \Delta + e^2 \JRConst \cdot \log k.
    % \end{equation*}
    
    Then, the condition \eqref{eq:DiagFrobCondition}:
    \begin{equation*}
        \exists x \in \RR_{\geq 0}^n,\, x \geq t \cdot \BUnit\colon \quad b = A x
    \end{equation*}
    implies that there exists an integer feasible solution of the system \eqref{eq:Sys-SF}, which can be found by an $2^{k} \cdot \poly(\inputsize)$-time algorithm.
\end{restatable}

% {\color{blue}

% As a lower bound, we present the following proposition concerning only $1 \times n$ matrices, meaning the value of the parameter $k$ is $1$. This example demonstrates that the linear dependence on $\Delta$ is tight.

% \begin{restatable}{proposition}{LBDiagFrob}\label{prop:LBDiagFrob}
%     There exists $1 \times n$ integer matrix $A$ such that $\DiagFrob(A) \geq $
% \end{restatable}

% The proofs of our main results are based on two ingredients: properties of the Gomory's corner polyhedron relaxation and results from discrepancy theory.  Furthermore, we consider a generalization of the diagonal Frobenius number, which imposes a condition on the slacks-vector of an (more general) integer linear programming problem in the canonical form. Following to this approach, we not only generalized the result, but also somewhat simplified the proof. A more detailed description is provided in the next Subsection. The proofs of \Cref{th:DiagFrob}, \Cref{th:PolyDiagFrob}, \Cref{th:ExpDiagFrob} and \Cref{prop:LBDiagFrob} could be found in \Cref{sec:standard_proofs}.

% }

\subsection{Diagonal Frobenius Number for Slacks}\label{sec:slacks_intro}

In this Subsection, we consider the \emph{integer linear feasibility problem in the canonical form} and define the corresponding version of the diagonal Frobenius number.
\begin{problem}\label{prbl:canonical_feasibility}
	Let $A \in \ZZ^{(n+k)\times n}$, $\rank(A) = n$, $b \in \ZZ^{n+k}$. \emph{The integer linear feasibility problem in the canonical form with $n+k$ constraints} can be formulated as the problem to find an integer feasible solution $z \in \ZZ^n$ of the following system:
	\begin{equation}
		\begin{cases}
			A x \leq b\\
			x \in \RR^n.
		\end{cases}\label{eq:Sys-CF}\tag{\(\SysCF\)}
	\end{equation}
%	Again, for simplicity, we assume that $\dim(\PC) = n$ for the corresponding polyhedra $\PC = \{x \in \RR^n \colon A x \leq b\}$.
\end{problem}

The corresponding version of the diagonal Frobenius number can be defined as follows.
\begin{definition}
	Corresponding to the system \eqref{eq:Sys-CF}, the \emph{diagonal Frobenius number for Slacks $\SlackFrob(A)$} is defined as the minimum $t \in \ZZ_{>0}$, such that, for each $b \in \ZZ^{n+k}$, the condition 
	\begin{equation}\label{eq:SlackFrobCondition}\tag{\(\SlackFrobCond(t)\)}
		\exists x \in \RR^n:\quad b - A x \geq t \cdot \BUnit
	\end{equation}
	implies that there exists an integer feasible solution $z \in \ZZ^n$ of the system \eqref{eq:Sys-CF}, that is
	\begin{equation*}
		\exists z \in \ZZ^n:\quad A z \leq b.
	\end{equation*}
\end{definition}

\begin{remark}\label{rm:both_formulations}
In this remark, we are going to justify the reason, why we consider the system \eqref{eq:Sys-CF} and the corresponding diagonal Frobenius number for slacks $\SlackFrob(A)$. The main reason for working with \eqref{eq:Sys-CF} is that it is geometrically more intuitive and even more general than \eqref{eq:Sys-SF}. The "geometric intuitive" part helps to simplify the proof idea. 

However, let us explain, why \eqref{eq:Sys-CF} is strictly more general than \eqref{eq:Sys-SF}. While these systems are mutually transformable, all known trivial transformations alter at least one of the key parameters $\bigl(k, d, \Delta(A)\bigr)$, where $d$ denotes the dimension of the corresponding polyhedra. The existence of a parameter-preserving transformation is a nontrivial question, resolved by \citet{OnCanonicalProblems_Grib}. This transformation will be explained in more detail in \Cref{sec:connection}. To make the both systems equivalent, one must augment the system \eqref{eq:Sys-SF} with additional constraints, described modulo a finite Abelian group, see \Cref{prbl:Sys-ModularSF}.

Thus, the described duality motivates to use of both formulations: while \eqref{eq:Sys-CF} offers greater generality, clear geometric intuition and proof simplification, \eqref{eq:Sys-SF} remains more prevalent in the ILP literature.
\end{remark}

Our main result with respect to $\SlackFrob(A)$ is stated in the following \Cref{th:SlackFrob}. The main result for $\DiagFrob(A)$ (\Cref{th:DiagFrob}) is a direct consequence of \Cref{th:SlackFrob} and a reduction between the systems \eqref{eq:Sys-SF} and \eqref{eq:Sys-CF}.
\begin{restatable}{theorem}{SlackFrobMainTh}\label{th:SlackFrob}
    Denote 
    \begin{gather*}
        t_1 = \Delta +  C \cdot \log k,\\
        t_2 = \Delta +  C \cdot \sqrt{\log k \cdot \log n},\\
        t = \min\{t_1,t_2\},
    \end{gather*}
    for a sufficiently large absolute constant $C$, whose exact value is not crucial for our purposes.
    % \begin{gather*}
    %     t_1 = \Delta +  \JRConst \cdot \log k,\\
    %     t_2 = \Delta +  \JRConst \cdot \sqrt{\log k \cdot \log n},\\
    %     t = \min\{t_1,t_2\}.
    % \end{gather*} 
    Then,
    \begin{equation*}
        \SlackFrob(A) \leq t.
    \end{equation*}

    % Assuming that $n \leq k$, the upper bound on $\SlackFrob(A)$ can be slightly refined to \begin{equation*}
    %     \SlackFrob(A) \leq t := \Delta +  C \cdot \sqrt{\log k \cdot \log n}.
    % \end{equation*}

    Additionally, assuming that a base $\BC$ with $\abs{\det A_{\BC}} = \Delta$ is known, the condition \eqref{eq:SlackFrobCondition}:
    \begin{equation*}
        \exists x \in \RR^n:\quad b - A x \geq t \cdot \BUnit
    \end{equation*}
    implies that there exists an integer feasible solution of the system \eqref{eq:Sys-CF}, which can be found by a polynomial-time algorithm.
\end{restatable}

Next, we state the generalizations (\Cref{th:PolySlackFrob} and \Cref{th:ExpSlackFrob}) of our results in \Cref{th:PolyDiagFrob} and \Cref{th:ExpDiagFrob} with respect to the diagonal Frobenius number for slacks. Again, \Cref{th:PolyDiagFrob} and \Cref{th:ExpDiagFrob} are consequences of \Cref{th:PolySlackFrob} and \Cref{th:ExpSlackFrob}. 

We recall that it is an NP-hard problem to construct a base $\BC$ of $A$ with $\abs{\det A_{\BC}} = \Delta(A)$. So, the bound of \Cref{th:SlackFrob} does not imply a polynomial-algorithm to construct a corresponding integer feasible solution of \eqref{eq:Sys-CF}. However, there exist weaker upper bounds for $\SlackFrob(A)$, which admit such polynomial-time algorithms. They are presented in the following Theorem.

\begin{restatable}{theorem}{PolynomialSlackFrobMainTh}\label{th:PolySlackFrob}
    Denote
    \begin{gather*}
        t_1 = \Delta + C_1 \cdot \begin{cases}
            \sqrt{k},\quad\text{for $k \leq n$},\\
            \sqrt{k} \cdot \log\bigl(\frac{2k}{n}\bigr),\quad\text{for $k \geq n$},
        \end{cases}\\
        t_2 = \Delta + C_2 \cdot \Delta \cdot \log k,\\
        t_3 = \Delta + C_2 \cdot \Delta \cdot \sqrt{\log k \cdot \log n},\\
        t = \min\{t_1,t_2,t_3\},
    \end{gather*}
    for sufficiently large absolute constants $C_1,C_2$, whose exact values are not crucial for our purposes.
    Then, the condition \eqref{eq:SlackFrobCondition}:
    \begin{equation*}
        \exists x \in \RR^n:\quad b - A x \geq t \cdot \BUnit
    \end{equation*}
    implies that there exists an integer feasible solution of the system \eqref{eq:Sys-CF}, which can be found by a polynomial-time algorithm.
\end{restatable}

Another interesting case is when $k$ is a constant or a slowly growing function, depending on the input size. If this situation occurs, there exists an upper bound on $\SlackFrob(A)$, which admits a $2^k \cdot \poly(\inputsize)$-time algorithm to find an integer feasible solution of \eqref{eq:Sys-CF}. The constant $C$ in the corresponding upper bound on $\SlackFrob(A)$ is $e^2$ times larger than the constant, we found in \Cref{th:SlackFrob}.

\begin{restatable}{theorem}{ExpSlackFrobMainTh}\label{th:ExpSlackFrob}
    Denote 
    \begin{gather*}
        t_1 = \Delta + C \cdot \log k,\\
        t_2 = \Delta + C \cdot \sqrt{\log k \cdot \log n},\\
        t = \min\{t_1,t_2\},
    \end{gather*}
    for a sufficiently large absolute constant $C$, whose exact value is not crucial for our purpose, where the constant $C$ in this Theorem is $e^2$ times larger than the constant in the Frobenius number $\SlackFrob(A)$, we found in \Cref{th:SlackFrob}.
    % \begin{equation*}
    %     t = \Delta + e^2 \JRConst \cdot \log k.
    % \end{equation*}
    
    Then, the condition \eqref{eq:SlackFrobCondition}:
    \begin{equation*}
        \exists x \in \RR^n:\quad b - A x \geq t \cdot \BUnit
    \end{equation*}
    implies that there exists an integer feasible solution of the system \eqref{eq:Sys-CF}, which can be found by a $2^{k} \cdot \poly(\inputsize)$-time algorithm.
\end{restatable}

As a lower bound, we present the following proposition. However, it only concerns $(n+1) \times n$ matrices, meaning the value of the parameter $k$ is $1$. 
% This example demonstrates that the linear dependence on $\Delta$ is tight.

\begin{restatable}{proposition}{LBSlackFrob}\label{prop:LBSlackFrob}
    There exists a matrix $A \in \ZZ^{(n+1) \times n}$ of rank $n$ such that $\SlackFrob(A) \geq (\Delta-2)/2$.
\end{restatable}

The proofs of \Cref{th:SlackFrob}, \Cref{th:PolySlackFrob}, \Cref{th:ExpSlackFrob} and \Cref{prop:LBSlackFrob} could be found in \Cref{sec:canonical_proofs}.

% \subsection{Structure of the Paper}\label{sec:pstracture}

\subsection{Complexity Model and Other Assumptions}\label{sec:assumptions}
All the algorithms that are considered in our work rely on the \emph{Word-RAM} computational model. In other words, we assume that additions, subtractions, multiplications, and divisions with rational numbers of the specified size, which is called the \emph{word size}, can be done in $O(1)$ time. In our work, we choose the word size to be equal to some fixed polynomial on $\lceil \log n \rceil + \lceil \log k \rceil + \lceil \log \alpha \rceil$, where $\alpha$ is the maximum absolute value of elements of $A$ and $b$ in the problem formulations.

% Sometimes, when it is important to specify the exact size of variables, we do this explicitly. Following \cite[Section~4.1]{KorteBook}, we define the bit-encoding size of an integer number $z \in \ZZ$ by the formula $\size(z) = 1 + \bigl\lceil\log_2\bigl(\,\abs{z}+1\bigr)\bigr\rceil$. For a rational number $r = p/q$, where $p$ and $q$ are coprime integers, the size is defined by the formula $\size(r) = \size(p) + \size(q)$.  

\begin{remark}\label{rm:GCD_assumption}
	Let us clarify the assumption $\Delta_{\gcd}(A) = 1$, which was done in \Cref{prbl:standard_feasibility}. Let us assume that $\Delta_{\gcd}(A) = d > 1$, and let us show that the original problem can be reduced to an equivalent new problem with $\Delta_{\gcd}(A') = 1$, using a polynomial-time reduction. 
	
	Let $A = P \cdot \bigl(S\,\BZero\bigr) \cdot Q$, where $\bigl(S\,\BZero\bigr) \in \ZZ^{k \times n}$, be the \emph{Smith Normal Form (the SNF, for short)} of $A$ and $P \in \ZZ^{k \times k}$, $Q \in \ZZ^{n \times n}$ be unimodular matrices. We multiply rows of the original system $A x = b,\, x \geq \BZero$ by the matrix $(P S)^{-1}$. After this step, the original system is transformed to the equivalent system $\bigl(I_{k \times k}\,\BZero\bigr) \cdot Q\,x = b^\prime$, $x \geq \BZero$. In the last formula, $b^\prime$ is integer, because in the opposite case the original system is integrally infeasible. Clearly, the matrix $\bigl(I_{k \times k}\,\BZero\bigr)$ is the SNF of $\bigl(I_{k \times k}\,\BZero\bigr)\,Q$, so its $\Delta_{\gcd}(\cdot)$ is equal to $1$. Finally, note that the computation of the SNF is a polynomial-time solvable problem, see \Cref{sec:Smith}. 
\end{remark}

\subsection{Other Related Work}

When the parameter $\Delta$ is bounded, the polyhedra, defined by systems $A x \leq b$, have many interesting properties in algorithmic perspective. Such polyhedra are also known under the name \emph{$\Delta$-modular polyhedra}. 

According to \citet{BimodularStrong}, when $\Delta \leq 2$, integer programming over $\Delta$-modular polyhedra can be solved, using a strongly polynomial-time algorithm. This advancement, built upon an earlier research by \citet{BimodularVert}, which laid the groundwork by characterizing key structural features of these polyhedra and demonstrating that the integer feasibility problem for such systems is decidable in polynomial time.

The work \citet{TwoNonZerosStrong} further showed that, for any fixed $\Delta$ and under the assumption that matrix $A$ contains no more than two nonzero entries per row, the corresponding integer linear program (ILP, for short) admits a strongly polynomial-time solution. Earlier a less general result has been established by \citet{AZ}, who proved that ILPs with a ${0,1}$-matrix $A$, having at most two non-zeros per row and a fixed value of $\Delta\binom{\BUnit^\top}{A}$, can be solved in linear time.

However, the computational complexity of ILP remains open for $\Delta = 3$ and arbitrary matrices $A$. Moreover, as shown by \citet{StableSetHardness}, unless the Exponential Time Hypothesis (ETH, for short) fails, there are no polynomial-time algorithms for ILP problems, where $\Delta = \Omega(n^\varepsilon)$, for any $\varepsilon > 0$.

Significant advances have been achieved in the analysis of $\Delta$-modular polyhedra, described by a system \eqref{eq:Sys-CF} with $n + k$ facets, where the number of facets equals the number of constraints, and those given by \eqref{eq:Sys-SF} with co-dimension $k$, under the assumption that $k$ is bounded. For this family of polyhedra, a number of computational results is known:
\begin{itemize}
    \item The integer linear programming problem and the integer feasibility problem can be solved in
    \begin{gather*}
        O(\log k)^{2k} \cdot \Delta^2 / 2^{\Omega(\sqrt{\log \Delta})},\\
        O(\log k)^{k} \cdot \Delta \cdot (\log \Delta)^3
    \end{gather*}
    arithmetic operations, respectively \cite{Gribanov_fixed_codim}.

    \item The number of integer points $\abs{\PC \cap \ZZ^n}$ can be calculated within
    \[
        O(n/k)^{2k} \cdot n^3 \cdot \Delta^3
    \]
    operations \cite{SparseILP_Gribanov,HyperAvoiding_NonHomo}. A parameterized version of the counting problem is given by \cite{Parametric_Counting_Grib}.

    \item All vertices of the integer hull $\conv(\PC \cap \ZZ^n)$ can be listed, using
    \[
        (k \cdot n \cdot \log \Delta)^{O(k + \log \Delta)}
    \]
    operations \cite{OnCanonicalProblems_Grib}.

    \item In the case of $\Delta$-modular simplices, their width is computable in $\poly(\Delta,n)$ time, and their unimodular equivalence classes can be enumerated by a polynomial-time algorithm, when $\Delta$ is fixed \cite{Width_Grib,WidthConv_Grib,SimplexEquiv_Gribanov}.
\end{itemize}

\section{Preliminaries}\label{sec:prelim}

\subsection{List of Notations}\label{sec:notation}

Let $A \in \RR^{k \times n}$. We will use the following notations:
\begin{itemize}
    \item $A_{i\,j}$ is the $(i,j)$-th entry of $A$;
    \item $A_{i\,*}$ is $i$-th row vector of $A$;
    \item $A_{*\,j}$ is $j$-th column vector of $A$;
    \item $A_{\IC \JC}$ is the sub-matrix of $A$, consisting of rows and columns, indexed by $\IC$ and $\JC$, respectively;
    \item Replacing $\IC$ or $\JC$ with $*$, selects all rows or columns, respectively;
    \item When unambiguous, we abbreviate $A_{\IC*}$ as $A_{\IC}$ and $A_{*\JC}$ as $A_{\JC}$.
\end{itemize}
In our work, we will often use the shorthand notation $\Delta$ to denote $\Delta(A)$.

For a matrix $A \in \RR^{k \times n}$, denote
\begin{gather*}
    \linh(A) = \left\{A x \colon x \in \RR^n\right\},\\
    \inth(A) = \left\{A x \colon x \in \ZZ^n\right\},\\
    \cone(A) = \left\{A x \colon x \in \RR_{\geq 0}^n\right\},\\
    \cone_{\ZZ}(A) = \left\{A x \colon x \in \ZZ_{\geq 0}^n\right\}.
\end{gather*}

For a matrix $A \in \RR^{k \times n}$, vectors $b \in \ZZ^k$ and $x \in \ZZ^n$, and a diagonal matrix $S \in \ZZ^{k \times k}$, the notation
\begin{equation*}
    A x \equiv b \pmod{S \cdot \ZZ^n}
\end{equation*} denotes that, for each $i \in \intint{k}$, there exists $z \in \ZZ$, such that $A_{i *} x = b_i + S_{i i} z$.

\subsection{The Smith and Hermite Normal Forms}\label{sec:Smith}

For any non-degenerate $A \in \ZZ^{n \times n}$, there exist \emph{unimodular} nondegenerate matrices $P,Q \in \ZZ^{n \times n}$, and $Q \in \ZZ^{n \times n}$, such that
\begin{equation*}
S = P A Q = \diag(s_1,s_2, \dots, s_n)
\end{equation*} with each $s_i \geq 1$ and $s_i \mid s_{i+1}$, for $i \in \intint{n-1}$. The matrix $S$ is called the \emph{Smith Normal Form of $A$} (or, shortly, the SNF of $A$). 

It is known that $\prod_{i = 1}^{k} s_{i} = \Delta_{\gcd}(A,k)$, for each $k \in \intint n$, we recall that $\Delta_{\gcd}(A,k)$ denotes the greatest common divisor of all the $k \times k$ sub-determinants of $A$. Thus, setting $\Delta_{\gcd}(A,0) = 1$, any of $s_i$ are uniquely defined by the formula $s_i = \Delta_{\gcd}(A,i)/\Delta_{\gcd}(A,i-1)$. 

Another useful and important matrix form is the \emph{Hermite Normal Form}. There exists a unimodular matrix $Q \in \ZZ^{n \times n}$, such that $A = H Q$, where $H \in \ZZ_{\geq 0}^{n \times n}$ is a lower-triangular matrix, such that $0 \leq H_{i j} < H_{i i}$, for any $i \in \intint n$ and $j \in \intint{i-1}$. The matrix $H$ is called the \emph{Hermite Normal Form} (or, shortly, the HNF) of the matrix $A$. 
% Additionally, it was shown in \cite{FPT_Grib} that $\|B\|_{\max} \leq \Delta(A)$ and, consequently, $\|\binom{H}{B}\|_{\max} \leq \Delta(A)$.

% (see, for example, \cite{Schrijver,HNFOptAlg}
Near-optimal polynomial-time algorithms for constructing the SNF and HNF of $A$ are given in \citet{SNFOptAlg,FastPSQDecomp,HNFOptAlg,StorjohannDoc}. The great survey about the SNF and other canonical matrix forms under principal ideal rings, such as the Howell Form, Hermite form, Frobenius form etc., can be found in \citet{StorjohannDoc}.

\subsection{Discrepancy Theory}\label{sec:prediscrep}

As it was noted before, we employ tools from discrepancy theory to prove our main results. Below, we provide a brief list of the required results and definitions.

\begin{definition}
For a matrix $A \in \RR^{k \times n}$, we recall the definitions of its \emph{discrepancy and hereditary discrepancy}:
\begin{gather*}
\disc(A) = \min_{z \in \{-1,\, 1\}^n} \norm{A z}_\infty = 2 \cdot  \min_{z \in \{0,\, 1\}^n} \norm{A (z - 1/2 \cdot \BUnit)}_\infty,\\
\herdisc(A) = \max_{\IC \subset \intint n} \disc(A_{* \IC}).
\end{gather*}
\end{definition}

% For our analysis, we require the following fundamental bounds on the hereditary discrepancy $\herdisc(A)$. 
The seminal result by \citet{SixDeviations_Spencer} establishes that, for any matrix $A \in \RR^{k \times n}$ with $k \geq n$, we have
\begin{equation}\label{eq:SixDeviations}
    \disc(A) = O\left( \Delta_1(A) \cdot \sqrt{n \cdot \log \bigl(\frac{2 k}{n}\bigr)} \right).
\end{equation}
The important matrix characteristic, that is closely related to $\herdisc(A)$, is $\detlb(A)$. According to \citet{HerDisc}, it can be defined as follows:
$$
\detlb(A) = \max\limits_{t \in \intint k} \sqrt[t]{\Delta_t(A)},
$$ and it was shown by \cite{HerDisc} that 
$
\herdisc(A) \geq \detlb(A)
$. It was shown by \citet{DiscDetBound} that $\detlb(A)$ can be used to produce tight upper bounds on $\herdisc(A)$. The result of Matou\v{s}ek was improved by \citet{TightDiscDetBound} as follows:
\begin{equation}\label{eq:DiscDetBound}
\disc(A) = O\left( \detlb(A) \cdot \sqrt{\log k \cdot \log n} \right).
\end{equation} 
% where $\JRConst$ denotes an implicit constant from the proof by \citet{TightDiscDetBound}.

Additionally, we will need the following important property, concerning discrepancy of matrices $A \in \RR^{k \times n}$, when $k \leq n$:
\begin{lemma}[{\citet[Corollary 13.3.3]{AlonSpencerBook} }]\label{DiscLowk_lm}
    Suppose that $\disc(A_{* \IC}) \leq H$, for every subset $\IC \in \intint n$ with $\abs{\IC} \leq k$. Then, $\disc(A) \leq 2 H$.
\end{lemma}
Originally, this statement was proved the only for discrepancy of hypergraphs. However, it is straightforward to see from the original proof that it extends to matrices as well.
In the assumption $k \leq n$, combining Lemma \ref{DiscLowk_lm} with
the upper bounds \eqref{eq:SixDeviations} and \eqref{eq:DiscDetBound}, 
% the upper bounds \eqref{SixDeviations_eq} and \eqref{DiscDetBound_eq}, 
we get
\begin{gather}
    \herdisc(A) = O\left(\log k \cdot \detlb(A)\right), \label{eq:DiscDetBoundReduced}\\
    \herdisc(A) = O\bigl(\sqrt{k} \cdot \Delta_1(A)\bigr). \label{eq:SixDeviationsReduced}
\end{gather}

An important application of discrepancy theory lies in constructing efficient rounding procedures to obtain integer solutions of linear equation systems. The rounding is considered successful if the rounded solution does not cause significant fluctuations in the right-hand side of the system. The original definition of $\disc(A)$ can be understood as a way to best approximate the vector $1/2 \cdot \BUnit$. The notion of $\lindisc(A)$ by \citet{HerDisc} allows to work with arbitrary vectors in $[0,1]^n$:
\begin{definition}
    For $A \in \RR^{k \times n}$, we recall the definitions of its \emph{linear discrepancy} and \emph{hereditary linear discrepancy}:
    \begin{gather*}
        \lindisc(A) = 2 \cdot \max\limits_{c \in [0,1]^n} \min\limits_{z \in \{0,1\}^n} \norm{A(z - c)}_\infty,\\
        \herlindisc(A) = \max\limits_{\IC \subseteq \intint n} \lindisc(A_{* \IC}).
    \end{gather*}
\end{definition}

By \citet[Corollary 1]{HerDisc}, we have
\begin{equation}\label{eq:lindisc_vs_herdisc}
    \lindisc(A) \leq \herlindisc(A) \leq \herdisc(A).
\end{equation}
Using \eqref{eq:lindisc_vs_herdisc}, the definition of $\lindisc(A)$ can be reformulated as a Lemma for rounding solutions of linear systems as follows:
\begin{lemma}\label{lm:Rounding01}
    For each $x \in [0,1]^n$, there exists $z \in \{0,1\}^n$, such that 
    \begin{equation*}
        \norm{Ax - Az}_\infty \leq \herdisc(A).
    \end{equation*}
\end{lemma}

\Cref{eq:lindisc_vs_herdisc} can be easily rewritten to handle vectors from $\RR_{\geq 0}^n$, see excellent lecture notes by \citet[Lecture 5]{DiscrepancyLections_Nikolov}:
\begin{lemma}\label{lm:Rounding0inf}
    For each $x \in \RR_{\geq 0}^n$, there exists $z \in \ZZ_{\geq 0}^n$, such that 
    \begin{equation*}
        \norm{Ax - Az}_\infty \leq \herdisc(A).
    \end{equation*}
\end{lemma}

\section{Connection Between Systems in the Canonical and Standard Forms
}\label{sec:connection}

In this Section, we describe a non-trivial connection between the systems \eqref{eq:Sys-CF} and \eqref{eq:Sys-SF}. We will survey the corresponding results by \citet{OnCanonicalProblems_Grib}. To make the exposition in a greater generality, we will consider the generalized optimization variants of the integer feasibility problems \Cref{prbl:standard_feasibility} and \Cref{prbl:canonical_feasibility}:  

\begin{problem}\label{prbl:ILP-SF}
Let $A \in \ZZ^{k\times n}$, $\rank(A) = k$, $c \in \ZZ^n$, $b \in \ZZ^{k}$. 
Assume additionally that all the $k \times k$ sub-determinants of $A$ are co-prime, where the clarification of this is given in Remark \ref{rm:GCD_assumption}.
\emph{The ILP problem in the standard form of co-dimension $k$} is formulated as follows:
\begin{align}
    &c^\top x \to \max \notag\\
    &\begin{cases}
    A x = b\\
    x \in \ZZ_{\geq 0}^n.
    \end{cases}\label{eq:ILP-SF}\tag{\(\ILPSF\)}
\end{align}
% For simplicity, we assume that $\dim(\PC) = n-k$ for the corresponding polyhedra $\PC = \{x \in \RR_{\geq 0}^n \colon A x = b\}$.
\end{problem}

\begin{problem}\label{prbl:ILP-CF}
Let $A \in \ZZ^{(n + k)\times n}$, $\rank(A) = n$, $c \in \ZZ^n$, $b \in \ZZ^{n+k}$. \emph{The ILP problem in the canonical form with $n+k$ constraints} can be formulated as follows:
\begin{align}
    &c^\top x \to \max \notag\\
    &\begin{cases}
    A x \leq b\\
    x \in \ZZ^n.
    \end{cases}\label{eq:ILP-CF}\tag{\(\ILPCF\)}
\end{align}
% Again, for simplicity, we assume that $\dim(\PC) = n$ for the corresponding polyhedra $\PC = \{x \in \RR^n \colon A x \leq b\}$.
\end{problem}

As it was briefly noted in \Cref{rm:both_formulations}, the formulation \eqref{eq:ILP-CF} is clearer from the geometric point of view, but it can be easily transformed to \eqref{eq:ILP-SF}, introducing some new integer variables. However, this straightforward reduction has a downside: It changes the value $k$ and the dimension of the corresponding polyhedra.

A more sophisticated reduction that preserves the parameters $k$, $\Delta$, and the dimension of the corresponding polyhedra is described by \citet{OnCanonicalProblems_Grib}. It connects the problem \eqref{eq:ILP-CF} with the equivalent problem, called the \emph{ILP problem in the standard form with modular constraints}, which strictly generalizes the problem \eqref{eq:ILP-SF}.
\begin{problem}\label{prbl:Sys-ModularSF}
    Let $A \in \ZZ^{k \times n}$ and $G \in \ZZ^{(n-k)\times n}$, such that $\binom{A}{G}$ is an integer non-degenerate $n \times n$ unimodular matrix. Additionally, let $S \in \ZZ^{(n-k)\times(n-k)}$ be a matrix, reduced to the Smith Normal Form (SNF, for short), $g \in \ZZ^{n-k}$, $b \in \ZZ^k$, $c \in \ZZ^n$. The \emph{ILP problem in the standard form of co-dimension $k$ with modular constraints} is formulated as follows:
    \begin{align}
        & c^\top x \to \max\notag\\
        & \begin{cases}
            A x = b\\
            G x \equiv g \pmod{S \cdot \ZZ^n}\\
            x \in \ZZ^n_{\geq 0}.
        \end{cases}\tag{$\ModularILPSF$}\label{eq:MOD-ILP-SF}
    \end{align}
    Here, the notation $G x \equiv g \pmod{S \cdot \ZZ^n}$ denotes that, for each $i \in \intint{(n-k)}$, there exists $z \in \ZZ$, such that $G_{i *} x = g_i + S_{i i} z$.
\end{problem}

Therefore, the problem \eqref{eq:ILP-CF} is strictly more general, since each exemplar of the \eqref{eq:ILP-CF} problem can be reduced to an exemplar of the \eqref{eq:MOD-ILP-SF} problem. If the equipped matrix $G$ has a non-trivial structure, such a problem can not be represented by \eqref{eq:ILP-SF}-type problems, if we want to preserve the parameters $k$, $\Delta$, and the dimension, see \citet[Remark~4]{OnCanonicalProblems_Grib} for the corresponding example.

Let us recall the formal description of the outlined reduction. It is given in the following Lemmas:
\begin{lemma}[{\citet[Lemma 4]{OnCanonicalProblems_Grib}}]\label{lm:ILPCF_to_ILPSF}
    For any instance of the \eqref{eq:ILP-CF} problem, there exists an equivalent instance of the \eqref{eq:MOD-ILP-SF} problem
    \begin{align*}
        & \hat c^\top x \to \min\\
        & \begin{cases}
            \hat A x = \hat b\\
            G x \equiv g \pmod{S \cdot \ZZ^n}\\
            x \in \ZZ_{\geq 0}^{n+k},
        \end{cases}
    \end{align*}
    with $\hat A \in \ZZ^{k \times (k+n)}$, $\rank(\hat A) = k$, $\hat b \in \ZZ^k$, $\hat c \in \ZZ^{n+k}$, $G \in \ZZ^{n \times (n+k)}$, $g \in \ZZ^n$, $S \in \ZZ^{n \times n}$. Moreover, the following properties hold:
    \begin{enumerate}
        \item $\hat A \cdot A = \BZero_{k \times n}$, $\Delta(\hat A) = \Delta(A)/\Delta_{\gcd}(A)$;
        \item $\abs{\det(S)} = \Delta_{\gcd}(A)$;
        \item There exists a bijection between rank-order sub-determinants of $A$ and $\hat A$;
        \item The map $\hat x = b - A x$ is a bijection between integer solutions of both problems;
        \item If the original relaxed LP problem is bounded, then we can assume that $\hat c \geq \BZero$;
        \item The reduction is not harder than the computation of the SNF of $A$.
    \end{enumerate}
\end{lemma}

\begin{lemma}[{\citet[Lemma 5]{OnCanonicalProblems_Grib}}]\label{lm:ILPSF_to_ILPCF}
    For any instance of the \eqref{eq:MOD-ILP-SF} problem, there exists an equivalent instance of the \eqref{eq:ILP-CF} problem
    \begin{align*}
        &\hat c^\top x \to \max\\
        &\begin{cases}
            \hat A x \leq \hat b\\
            x \in \ZZ^d
        \end{cases}
    \end{align*}
    with $d = n - k$, $\hat A \in \ZZ^{(d+k)\times d}$, $\rank(\hat A) = d$, $\hat c \in \ZZ^d$, and $b \in \ZZ^{d+k}$. Moreover, the following properties hold:
    \begin{enumerate}
        \item $A \cdot \hat A = \BZero_{k \times d}$, $\Delta(\hat A) = \Delta(A) \cdot \abs{\det(S)}$;
        \item $\Delta_{\gcd}(\hat A) = \abs{\det(S)}$;
        \item There exists a bijection between rank-order sub-determinants of $A$ and $\hat A$;
        \item The map $x = \hat b - \hat A \hat x$ is a bijection between integer solutions of both problems;
        \item The reduction is not harder than the inversion of an integer unimodular $n \times n$ matrix $\binom{A}{G}$.
    \end{enumerate}
\end{lemma}
 
\begin{remark}\label{rm:StandardFormReduction}
In this remark, we justify that \eqref{eq:ILP-SF} is a special case of \eqref{eq:MOD-ILP-SF}. By \Cref{lm:ILPSF_to_ILPCF}, the latter means that \eqref{eq:ILP-SF} is a special case of \eqref{eq:ILP-CF} modulo a polynomial-time reduction procedure.

By \Cref{rm:GCD_assumption}, we can assume that $\Delta_{\gcd}(A) = 1$, 
% Indeed, let $A = P\, \bigl(S\,\BZero\bigr)\, Q$, where $\bigl(S\,\BZero\bigr) \in \ZZ^{m \times n}$ be the SNF of $A$, and $P \in \ZZ^{m \times m}$, $Q \in \ZZ^{n \times n}$ be unimodular matrices. We multiply rows of the original system $A x = b,\, x \geq \BZero$ by the matrix $(P S)^{-1}$. After this step, the original system is transformed to the equivalent system $\bigl(I_{m \times m}\,\BZero\bigr)\,Q\,x = b^\prime$, $x \geq \BZero$. In the last formula $b^\prime \in \ZZ^m$, because in the opposite case the system is integrally infeasible. Clearly, the matrix $\bigl(I_{m \times m}\,\BZero\bigr)$ is the SNF of $\bigl(I_{m \times m}\,\BZero\bigr)\,Q$, so its $\Delta_{\gcd}(\cdot)$ is equal to $1$.
which means that the columns of $A^\top$ form a primitive basis of some sub-lattice of $\ZZ^n$. Hence, it can be extended to a full basis of $\ZZ^n$. Let the columns of $G^\top \in \ZZ^{(n-m) \times n}$ form this extension, which can be constructed by a polynomial-time algorithm. Consequently, $\dbinom{A}{G}$ is a $n \times n$ integral non-degenerate unimodular matrix. Thus, the \eqref{eq:ILP-SF} problem in the \eqref{eq:MOD-ILP-SF}-form, where $I$ is the $(n-m) \times (n-m)$ identity matrix, is:
\begin{gather*}
    c^\top x \to \min\\
    \begin{cases}
    A x = b\\
    G x \equiv \BZero \pmod{I \cdot \ZZ^n}\\
    x \in \ZZ^n_{\geq 0}.
    \end{cases}
\end{gather*}
By \Cref{lm:ILPSF_to_ILPCF}, this system can be reduced to the \eqref{eq:ILP-CF} problem, using a polynomial-time algorithm.
\end{remark}

\section{The Gomory's Corner Polyhedron Relaxation}\label{sec:Gomory_relaxation}

The \emph{Gomory's corner polyhedron relaxation} was defined by \citet{GomoryRelation}, see also \citet{GomoryIntegerFaces,GomoryCombinatorialPoly} and the excellent book, which covers the topic \citet{HuBook}. The original form of the Gomory's construction considers the problem \eqref{eq:Sys-SF} and an arbitrary optimal base $\BC$ of the corresponding LP relaxation. If we rewrite the problem \eqref{eq:ILP-SF} in the following form:
\begin{align*}
    &c_{\BC}^\top x_{\BC} + c_{\NotBC}^\top x_{\NotBC} \to \min\\
    &\begin{cases}
        A_{\BC} x_{\BC} + A_{\NotBC} x_{\NotBC} = b\\
        x_{\BC} \in \ZZ_{\geq 0}^k,\, x_{\NotBC} \in \ZZ_{\geq 0}^{n-k},
    \end{cases}
\end{align*}
then its \emph{Gomory's corner polyhedron relaxation} can be achieved, relaxing the constraint $x_{\BC} \geq \BZero$:
\begin{align*}
    &c_{\BC}^\top x_{\BC} + c_{\NotBC}^\top x_{\NotBC} \to \min\\
    &\begin{cases}
        A_{\BC} x_{\BC} + A_{\NotBC} x_{\NotBC} = b\\
        x_{\BC} \in \ZZ^k,\, x_{\NotBC} \in \ZZ_{\geq 0}^{n-k}.
    \end{cases}
\end{align*}

To the best of our knowledge, another view on the Gomory's relaxation from the standpoint of problems \eqref{eq:ILP-CF} in the canonical form, first appeared in \citet[Paragraph~3.3, p.~42--43]{BlueBook}. The \emph{Gomory's corner polyhedron relaxation} with respect to the problem \eqref{eq:ILP-CF} can be stated in a very simple form:
\begin{align*}
    &c^\top x \to \max\\
    &\begin{cases}
        A_{\BC} x \leq b_{\BC}\\
        x \in \ZZ^n,
    \end{cases}
\end{align*}
where $\BC$ is an optimal base of the corresponding LP relaxation. One of the theses, put forward in the works \citet{BlueBook,OnCanonicalProblems_Grib}, is that, from the perspective of the problems in the canonical form \big(the problem \eqref{eq:ILP-CF}\big), the structure of integer points within the Gomory's relaxation is more transparent and accessible. By \citet{OnCanonicalProblems_Grib}, the class of problems, whose optimal solution coincides with an optimal solution of the Gomory's relaxation. The work by \citet{IntegralityNumber} contains a condition, when the problem \eqref{eq:ILP-CF} could become local. However, we again cite \cite{OnCanonicalProblems_Grib}, which gives a slightly tighter condition. Additionally, we note that the work by \citet[Lemma 4]{Width_Grib} implicitly presents a polynomial-time algorithm to find a feasible integer solution of a local problem\footnote{Note that we are referring specifically to a feasible solution, not an optimal one. Finding an optimal solution to the relaxation is a more difficult problem; however, we are not aware of any references claiming it to be NP-hard.}.

In the following Theorem, we unify the locality condition for the integer feasibility problems in \eqref{eq:Sys-CF} (by \cite{OnCanonicalProblems_Grib}), and the corresponding polynomial-time algorithm to find an integer feasible solution of a local problem (by \cite{Width_Grib}). For the completeness and clarity, we provide a complete proof.
\begin{theorem}[{\citet{OnCanonicalProblems_Grib} with \citet[Lemma 4]{Width_Grib}}]\label{th:canonical_Gomory}
	Let $\BC$ be a feasible base, corresponding to the system \eqref{eq:Sys-CF}, and let $v_{\BC} = A_{\BC}^{-1} b_{\BC}$ be the corresponding vertex-solution. Denote $\Delta = \Delta(A)$. If the condition 
\begin{equation}\label{eq:gomory_slack_condition}
		b_{\NotBC} - A_{\NotBC} v_{\BC} \geq (\Delta - 1) \cdot \BUnit
	\end{equation}
	is satisfied, then the system \eqref{eq:Sys-CF} has an integer feasible solution that can be found by a polynomial-time algorithm. 
\end{theorem}
\begin{proof}
	Let us consider the Gomory's corner polyhedron relaxation with respect to the base $\BC$, which can be written just by the subsystem 
\begin{equation}\label{eq:square_subs}
		\begin{cases}
			A_{\BC} x \leq b_{\BC}\\
			x \in \ZZ^n
		\end{cases}
	\end{equation}
	of the original system \eqref{eq:Sys-CF}. Denote the corresponding slack variables by a vector $y$, that is $y = b_{\BC} - A_{\BC} x$. We claim that \eqref{eq:square_subs} has an integer feasible solution $z$ with the corresponding slack vector $y$, satisfying
	\begin{equation}\label{eq:gomory_solution_slack}
		\norm{y}_1 \leq \delta_{\BC}-1,
	\end{equation} 
	denoting $\delta_{\BC} = \abs{\det A_{\BC}}$. Moreover, we claim that $z$ can be found by a polynomial-time algorithm. 
	
	Let $A_{\BC} = H Q^{-1}$, where $H \in \ZZ^{n \times n}$ be the HNF of $A_{\BC}$ and $Q \in \ZZ^{n \times n}$ be unimodular. Using the map $x \to Q x$, the system \eqref{eq:square_subs} can be rewritten to
	\begin{equation*}
		\begin{cases}
			H x \leq b_{\BC}\\
			x \in \ZZ^n.
		\end{cases}
	\end{equation*}
	With respect to the slack variables $y$, we have $y = b_{\BC} - H x$. Since $H$ is lower triangular, it is easy to see that we can choose $x \in \ZZ^n$, such that $y_i \in \intint[0]{H_{i i}-1}$, for each $i \in \intint n$. Hence, $\norm{y}_1 \leq \sum_i (H_{i i}-1)$. Since $\prod_i H_{i i} = \delta_{\BC}$, we get $\norm{y}_1 \leq \delta_{\BC}-1$. Applying the inverse map $x \to Q^{-1} x$ and denoting $z := x$, we conclude that the desired solution $z$ of \eqref{eq:square_subs} has been found. The provided calculation is not harder than the computation of the HNF, which can be done by a polynomial-time algorithm and proves the claim.
	
	Now, let us show that the condition \eqref{eq:gomory_slack_condition} implies that $z$ is an integer feasible solution of \eqref{eq:Sys-CF}. In other words, we need to check that $A_{\NotBC} z \leq b_{\NotBC}$. Note that all the elements of the matrix $A_{\NotBC} A_{\BC}^{-1}$ are bounded by $\Delta/\delta_{\BC}$ in the absolute value. Therefore,
	\begin{multline*}
		A_{\NotBC} z = A_{\NotBC} ( A_{\BC}^{-1} (b_{\BC} - y) ) = A_{\NotBC} v_{\BC} - A_{\NotBC} A_{\BC}^{-1} y \overset{\text{by \eqref{eq:gomory_solution_slack}}}{\leq} \\ A_{\NotBC} v_{\BC} - \Delta \frac{\delta_{\BC}-1}{\delta_{\BC}} \cdot \BUnit \overset{\text{by \eqref{eq:gomory_slack_condition}}}{\leq} b_{\NotBC}
	\end{multline*}
	Here, we have also used that $\frac{\delta_{\BC}-1}{\delta_{\BC}}$ is monotone increasing, and consequently $\Delta \frac{\delta_{\BC}-1}{\delta_{\BC}} \leq \Delta-1$. The proof follows.
	
%	Let $A_{\BC} = P^{-1} S Q^{-1}$, where $S \in \ZZ^{n \times n}$ be the SNF of $A_{\BC}$ and $P,Q \in \ZZ^{n \times n}$ are unimodular. Introducing the slack variables $y \in \ZZ_{\geq 0}^n$ and after the map $x \to Q x$, the system \eqref{eq:square_subs} can be rewritten as
%	\begin{equation*}
%		\begin{cases}
%			S x + P y = P b\\
%			x \in \ZZ^n, y \in \ZZ_{\geq 0}^n,
%		\end{cases}
%	\end{equation*}
%	which can be further rewritten as
%	\begin{equation*}
%		\begin{cases}
%			P y = P b \pmod{S\cdot\ZZ^n}\\
%			y \in \ZZ_{\geq 0}^n.
%		\end{cases}
%	\end{equation*}
\end{proof}

The following simple example shows that the condition \eqref{eq:gomory_slack_condition} of \Cref{th:canonical_Gomory} is tight:
\begin{example}\label{ex:gomory_is_tight}
    Let us consider the $n$-dimensional polyhedron $\PC \subseteq \RR^n$, defined by a system
    \begin{equation*}
    	\begin{cases}
    		B x \leq b\\
    		c^\top x \geq 1\\
    		x \in \RR^n,
    	\end{cases}
    \end{equation*}
    where $B$ is an $ n\times n$ diagonal matrix with $\diag(B) = (1,\dots,1,p)$, $c = (0,\dots,0,p)^\top$, and $b = (0,\dots,0,p-1)^\top$.

    Note that $\PC \cap \ZZ^n = \emptyset$ for $p \geq 2$. Indeed, from $B x \leq b$ we get $x_n \leq \frac{p-1}{p}$ and from $c^\top x \geq 1$ we get $x_n \geq 1/p$. Thus, $\frac{1}{p} \leq x_n \leq \frac{p-1}{p}$ and $\PC \cap \ZZ^n = \emptyset$.

    Consider now a vertex $v = B^{-1} b = (0,\dots,0,\frac{p-1}{p})^\top$. Simple calculations provide $c^\top v -1 = p-2$. 
\end{example}

% {\color{blue} 

% In the previous \Cref{ex:gomory_is_tight}, the greatest common divisor of the $n \times n$ subdeterminants of $\binom{B}{-c^\top}$ equals $p$. For further analysis, we construct another example with a slightly weaker bound on $c^\top v - 1$, but with coprime $n \times n$ subdeterminants.

% \begin{example}\label{ex:gomory_is_tight_coprime}
%     Let us consider the $n$-dimensional polyhedron $\PC \subseteq \RR^n$, defined by a system
%     \begin{equation*}
%     	\begin{cases}
%     		B x \leq b\\
%     		c^\top x \geq 1\\
%     		x \in \RR^n,
%     	\end{cases}
%     \end{equation*}
%     where $B$ is an $ n\times n$ diagonal matrix with $\diag(B) = (1,\dots,1,p)$, $c = (0,\dots,0,p-1)^\top$, and $b = (0,\dots,0,p-1)^\top$.

%     Note that $\PC \cap \ZZ^n = \emptyset$ for $p \geq 3$. Indeed, from $B x \leq b$ we get $x_n \leq \frac{p-1}{p}$ and from $c^\top x \geq 1$ we get $x_n \geq 1/(p-1)$. Thus, $\frac{1}{p-1} \leq x_n \leq \frac{p-1}{p}$ and $\PC \cap \ZZ^n = \emptyset$.

%     Consider now a vertex $v = B^{-1} b = (0,\dots,0,\frac{p-1}{p})^\top$. Simple calculations provide $c^\top v -1 = \frac{(p-1)^2}{p} - 1 = p -3 + 1/p$. 
% \end{example}
% }

For the sake of completeness, we should also cite an earlier result by \citet{DistributionsILP}, obtained for the special case of systems in the standard form, i.e., the problem \eqref{eq:ILP-SF}, see also \citet{OnCanonicalProblems_Grib} for an alternative proof. This result follows from \Cref{th:canonical_Gomory} by applying the reducibility between problems in the canonical and standard forms, described in \Cref{sec:connection}. We will present it in the form for integer feasibility problems in \eqref{eq:Sys-SF}, along with the polynomial-time algorithm, following from \cite[Lemma 4]{Width_Grib} and a reduction between the problems.
\begin{theorem}[{\citet{DistributionsILP} with \citet[Lemma 4]{Width_Grib}}]\label{th:standard_Gomory}
	Let $\BC$ be a feasible base, corresponding to the system \eqref{eq:Sys-SF}. If $A_{\BC}^{-1} b \geq (\Delta-1)\cdot \BUnit$, then the system has a feasible integer solution, which can be found by a polynomial-time algorithm.
\end{theorem}

\section{How to Construct a Sufficiently Good Base}

Let $A \in \ZZ^{k \times n}$ with $\rank A = k$ and denote $\Delta = \Delta(A)$. In this Subsection, we discuss some approaches to find a base $\BC$, which gives sufficiently good upper bounds on the value of $\max_{i \in \intint k}\bigl\{\Delta_i\bigl(M(\BC)\bigr)\bigl\}$, where we have denoted $M(\BC) = A_{\BC}^{-1} A_{\NotBC}$.
% To this end, we seek for a base $\BC$ of $A$, which will simultaneously minimize the values $\Delta_i\bigl(M(\BC)\bigr)$, for $i \in \intint k$, where $M(\BC) := A_{\BC}^{-1} \cdot A$. 
Taking $\BC$, such that $\abs{\det(A_{\BC})} = \Delta$, we get $\Delta_i\bigl(M(\BC)\bigr) = 1$, for all $i \in \intint k$. But it is an NP-hard problem to compute such a $\BC$. Instead, we will settle for an approximate solution that can be obtained by a polynomial-time algorithm. The following Theorem, due to A.~Nikolov, gives an asymptotically optimal approximation ratio. 
\begin{theorem}[A.~Nikolov \cite{LargestSimplex_Nikolov}]\label{th:maxdet_apr}
    There exists a deterministic polynomial-time algorithm that computes a base $\BC$ of $A$ with $\Delta / \abs{\det(A_{\BC})} \leq e^k$.
\end{theorem}

The next Lemma uses an algorithm by A.~Nikolov to compute a relatively good base $\BC$.
\begin{lemma}\label{lm:exp_subdet_search}
% Let $A \in \ZZ^{k \times n}$, $\rank(A) = k$, and $\Delta := \Delta(A)$. 
% For a base $\BC \subseteq \intint n$ of $A$, denote $M(\BC) := (A_{\BC})^{-1} \cdot A$. 
There exists a base $\BC$ of $A$, such that
\begin{enumerate}
    % \item $\Delta_1\bigl(M(\BC)\bigr) \leq 1$;
    
    \item for each $i \in \intint k$, $\Delta_i\bigl(M(\BC)\bigr) \leq e^{i+1}$;

    \item the base $\BC$ can be computed by an algorithm with the computational complexity bound
    $$
    O\bigl( k \cdot 2^k \cdot T_{\text{apr}} \bigr),
    $$ where $T_{\text{apr}}$ is the computational complexity of the algorithm in \Cref{th:maxdet_apr} with an input $A$ (writing the complexity bound, we make the additional assumption that the approximation problem is harder than the matrix inversion).
\end{enumerate}
\end{lemma}
\begin{proof}
    Initially, we compute a base $\BC$ with $\Delta\bigl(M(\BC)\bigr) \leq e^k$, using \Cref{th:maxdet_apr}. 
    Next, we repeatedly perform the following iterations:
    \begin{algorithmic}[1]
    \State $M \gets M(\BC)$
    \For{$\JC \subseteq \intint k$}
        \State $i \gets \abs{\JC}$
        \State using \Cref{th:maxdet_apr}, compute a base $\IC$ of $M_{\JC *}$, such that $\abs{\det(M_{\JC \IC})} \cdot e^i \geq \Delta_i(M_{\JC *})$
        \If{$\abs{\det(M_{\JC \IC})} > e$}
            \State $\BC \gets \BC \setminus \JC \cup \IC$
            \State \textbf{break}
        \EndIf
    \EndFor
    \end{algorithmic}
    Note that $\bigl(M(\BC)\bigr)_{\BC} = I$, where $I$ is the $k \times k$ identity matrix.
    Hence, if the condition $\abs{\det(M_{\JC \IC})} > e$ will be satisfied, for some $\IC$ and $\JC$, then the value of $\abs{\det(A_{\BC})}$ will grow at least by $e$. Therefore, since initially 
    \begin{equation*}
        e^{-k} \cdot \Delta(A) \leq \abs{\det(A_{\BC})} \leq \Delta(A),
    \end{equation*}
    it is sufficient to run the described procedure $k$ times. More precisely, we can stop at the moment, when the cycle in Line 2 will be completely finished without calling the $\textbf{break}$-operator of Line 7. After that the condition $\Delta_i\bigl(M(\BC)\bigr) \leq e^{i+1}$ will be satisfied, for all $i \in \intint k$. Clearly, the total computational complexity is bounded by $O(k \cdot 2^k \cdot T_{\text{apr}})$.
\end{proof}

The next Lemma gives weaker conditions on $M(\BC)$, using a polynomial-time algorithm. But, it only gives guaranties on  $\Delta_1\bigl(M(\BC)\bigr) \leq e$.
\begin{lemma}\label{lm:poly_subdet_search}
% In assumptions of Lemma \ref{lm:exp_subdet_search}, 
There exists a base $\BC$ of $A$, which can be computed by a polynomial-time algorithm, such that
\begin{equation*}
    \Delta_1\bigl(M(\BC)\bigr) \leq e.
\end{equation*}
% \begin{enumerate}
%     \item $\Delta_1\bigl(M(\BC)\bigr) \leq e$;

%     % \item the base $\BC$ can be computed by an algorithm with complexity bound
%     % $$
%     % O\bigl( k^{\omega-1} \cdot n + \log(\Delta) \cdot k \cdot n \bigr).
%     % $$

%     \item the base $\BC$ can be computed by a polynomial-time algorithm.
% \end{enumerate}
\end{lemma}
\begin{proof}
% First of all, let us choose any base $\BC$ of $A$, which can be done with $O(k^{\omega-1} \cdot n)$ operations.
Initially, choose any base $\BC$ of $A$, which can be done by a polynomial-time algorithm.
Similar to the proof of \Cref{lm:exp_subdet_search}, we repeatedly perform the following procedure:
\begin{algorithmic}[1]
    \State $M \gets M(\BC)$
    \For{$i \in \intint k ,\; j \in \intint n$}
        \If{$M_{i j} > e$}
            \State $\BC \gets \BC \setminus \{i\} \cup \{j\}$
            \State \textbf{break}
        \EndIf
    \EndFor
\end{algorithmic}
Clearly, it is sufficient to run the procedure at most $O(\log \Delta)$ times. After that the property $\Delta_1(M) \leq e$ will be satisfied. Since we need $O(k \cdot n)$ operations to scan over all the elements of $M(\BC)$ and the same number of operations to recompute $M(\BC)$, the complexity of a single step is bounded by $O(k \cdot n)$, which lead us to a polynomial-time algorithm.  
\end{proof}

% The following Theorem was proved in \citet{MinimalDistance_Veselov}. Other proofs can also be found in \citet{BlueBook,CountingFixedM,ABCModular}.
% \begin{theorem}[\citet{MinimalDistance_Veselov}]\label{th:perp_matricies}
% Let $A \in \ZZ^{n \times k}$, $B \in \ZZ^{n \times (n-k)}$, $\rank A = k$, $\rank B = n-k$, and $A^\top B = \BZero$. Then, for any $\BC \subseteq \intint n$, $\abs{\BC} = k$, and $\NotBC = \intint{n} \setminus \BC$, the following equality holds:
% $$
% \Delta_{\gcd}(B) \cdot \abs{\det A_{\BC *}} = \Delta_{\gcd}(A) \cdot \abs{\det B_{\NotBC  *}}.
% $$
% \end{theorem}

% \begin{remark}
% The result of this Theorem was strengthened in \citet{MinorsOfOrtMatricies}. Namely, it was shown that the matrices $A$ and $B$ have the same diagonal of their SNF's modulo of $\gcd$-like multipliers.
% \end{remark}

Finally, we provide an algorithm to find a sufficiently good base of an integer matrix $B$, which has a different size $n \times (n-k)$ and rank $(n-k)$. Due to the trivial bijection between the families of sets $\binom{\intint n}{k}$ and $\binom{\intint n}{n-k}$, it can be found by almost the same algorithm as in \Cref{lm:exp_subdet_search}.
\begin{lemma}\label{lm:exp_subdet_search_dual}
Let $B \in \ZZ^{n \times (n-k)}$ with $\rank(B) = n-k$. There exists a base $\BC$ of $B$, such that
\begin{enumerate}
    \item for each $i \in \intint{n-k}$, $\Delta_i\bigl(M(B)\bigr) \leq e^{i+1}$;

    \item the base $\BC$ can be computed by an algorithm with the computational complexity bound
    $$
    O\bigl( k \cdot 2^k \cdot T_{\text{apr}} \bigr),
    $$ where $T_{\text{apr}}$ is the computational complexity of the algorithm of Theorem \ref{th:maxdet_apr} with an input $A$. Writing the complexity bound, we make an additional assumption that the approximation problem is harder than the matrix inversion.
\end{enumerate}
\end{lemma}
\begin{proof}
    Let $\BC$ be an arbitrary base of $A$, and let us assume that $\BC = \intint{(n-k)}$. Then, we have $M(\BC) = A_{\BC}^{-1} A = \bigl(I\,C\bigr)$, where $I$ is $(n-k) \times (n-k)$ identity matrix and $C$ is the $(n-k)\times k$ matrix. According to this observation, we can use an algorithm, entirely similar to the one in the proof of \Cref{lm:exp_subdet_search}. To do this, we apply the approximation algorithm of \Cref{lm:poly_subdet_search} to all the column subsets of $C$; there are exactly $2^k$ of them. As in the proof of \Cref{lm:exp_subdet_search}, the number of iterations is bounded by $k$.
\end{proof}
% \begin{proof}
%     We seek for a matrix $A \in \ZZ^{n \times k}$ of rank $k$ such that $B^\top A = \BZero_{(n-k) \times k}$. Decompose $B^\top = P^{-1} \bigl(S\:\BZero_{(n-k) \times k}\bigr) Q^{-1}$, where $\bigl(S\:\BZero_{(n-k) \times k}\bigr)$ is the SNF of $B^\top$ and $P \in \ZZ^{(n-k)\times(n-k)}$, $Q \in \ZZ^{n \times n}$ are unimodular. From this we can see that the matrix $A$ can be constructed from the columns of $Q$ with indices $\intint[(n-k)+1]{n}$. Since $Q$ is unimodular, we have $\Delta_{\gcd}(A) = 1$.  

%     By \Cref{th:perp_matricies}, each base $\BC$ of $A$ bijectively corresponds to a base $\NotBC = \intint n \setminus \BC$ of $B$, and
%     \begin{equation*}
%         \abs{\det A_{\BC *}} = \frac{\abs{\det B_{\NotBC *}}}{\Delta_{\gcd}(B)}.
%     \end{equation*}
%     Using \Cref{lm:exp_subdet_search}, construct a base $\BC$ of $A$
% \end{proof}

\section{Proofs of the Main Results}\label{sec:proofs}

In this Section, we will present proofs of the main Theorems. \Cref{sec:standard_proofs} contains proofs of \Cref{th:SlackFrob}, \Cref{th:PolySlackFrob}, and \Cref{th:ExpSlackFrob}. \Cref{sec:canonical_proofs} contains proofs of \Cref{th:DiagFrob}, \Cref{th:PolyDiagFrob}, and \Cref{th:ExpDiagFrob}.

\subsection{Proofs with Respect to Systems in the Canonical Form}\label{sec:canonical_proofs}

First, we will prove a key Lemma, which connects the properties of the Gomory's corner polyhedron relaxation (\Cref{th:canonical_Gomory}) with tools of discrepancy theory (\Cref{sec:prediscrep}).
\begin{lemma}\label{lm:SlackFrob_via_Discr}
    Corresponding to the system \eqref{eq:Sys-CF}, let $\BC$ be a given base of $A$ and let $\gamma$ be an upper bound on $\herdisc(A_{\NotBC} A_{\BC}^{-1})$. Denote $\Delta = \Delta(A)$ and $t = \Delta-1 + \gamma$. 
    
    Then, if the condition \eqref{eq:SlackFrobCondition}:
    \begin{equation*}
        \exists x \in \RR^n:\quad b - A x \geq t \cdot \BUnit.
    \end{equation*}
    is satisfied, then there exists a solution of the system \eqref{eq:Sys-CF}, which can be found by a polynomial-time algorithm.
\end{lemma}
\begin{proof}
    Let $x \in \RR^n$ be a solution of \eqref{eq:Sys-CF}, satisfying
    \begin{equation*}
        b - A x \geq t \cdot \BUnit.
    \end{equation*}
    Denote
    \begin{gather*}
        y_{\BC} = b_{\BC} - A_{\BC} x \geq t \cdot \BUnit,\\
        y_{\NotBC} = b_{\NotBC} - A_{\NotBC} x \geq t \cdot \BUnit.
    \end{gather*}
    From these relations, we have $x = A_{\BC}^{-1} (b_{\BC} - y_{\BC})$ and
    \begin{multline*}
        y_{\NotBC} = b_{\NotBC} - A_{\NotBC} \bigl(A_{\BC}^{-1} (b_{\BC} - y_{\BC})\bigr) = \\
        b_{\NotBC} - A_{\NotBC} A_{\BC}^{-1} b_{\BC} + A_{\NotBC} A_{\BC}^{-1} y_{\BC}.
    \end{multline*}
    % where $v_{\BC} = A_{\BC}^{-1} b_{\BC}$ denotes a vertex of the polyhedron induced by $A x \leq b$ corresponding to a feasible base $\BC$.

    By \Cref{lm:Rounding0inf}, there exists $z \in \ZZ^n_{\geq 0}$, such that
    \begin{equation}\label{eq:alpha_y_disc_bound}
        \norm{A_{\NotBC} A_{\BC}^{-1} (y_{\BC} - z)}_\infty \leq \gamma.
    \end{equation}
    Denoting $\alpha = y_{\BC} - z$, $\hat b_{\BC} = b_{\BC} - z$, $\hat b_{\NotBC} = b_{\NotBC}$, and $\hat v_{\BC} = A_{\BC}^{-1} \hat b_{\BC}$, we get
    \begin{multline*}
        y_{\NotBC} = b_{\NotBC} - A_{\NotBC} A_{\BC}^{-1} b_{\BC} + A_{\NotBC} A_{\BC}^{-1} z + A_{\NotBC} A_{\BC}^{-1} \alpha = \\
        b_{\NotBC} - A_{\NotBC} A_{\BC}^{-1} \hat b_{\BC} + A_{\NotBC} A_{\BC}^{-1} \alpha = \\
        b_{\NotBC} - A_{\NotBC} \hat v_{\BC} + A_{\NotBC} A_{\BC}^{-1} \alpha.
    \end{multline*}
    
    Recalling the definition of $t$ and that $y_{\NotBC} \geq t \cdot \BUnit$, by \eqref{eq:alpha_y_disc_bound}, we have 
    \begin{equation}\label{eq:hatb_slack_condition}
        \hat b_{\NotBC} - A_{\NotBC} \hat v_{\BC} = y_{\NotBC} - A_{\NotBC} A_{\BC}^{-1} \alpha \geq (t - \gamma) \cdot \BUnit \geq (\Delta -1) \cdot \BUnit.
    \end{equation}
    
    Let us consider the system
    \begin{equation}\label{eq:hatb_canonical_system}
        \begin{cases}
            A x \leq \hat b,\\
            x \in \RR^n.
        \end{cases}
    \end{equation}
    By \eqref{eq:hatb_slack_condition}, $\hat v_{\BC}$ is a vertex of the polyhedra, defined by \eqref{eq:hatb_canonical_system}. Moreover, by \Cref{th:canonical_Gomory} and \eqref{eq:hatb_slack_condition}, the system \eqref{eq:hatb_canonical_system} admits a feasible integer solution $\hat x \in \ZZ^n$, which can be constructed by a polynomial-time algorithm. Finally, note that $\hat x$ is an integer feasible solution of the original system $A x \leq b$. Indeed, since $z \geq \BZero$, we have
    \begin{gather*}
        A_{\BC} \hat x \leq \hat b_{\BC} = b_{\BC} - z \leq b_{\BC},\\
        A_{\NotBC} \hat x \leq \hat b_{\NotBC} = b_{\NotBC},
    \end{gather*}
    which finishes the proof.
\end{proof}

Next, we present the proofs of \Cref{th:SlackFrob}, \Cref{th:PolySlackFrob}, and \Cref{th:ExpSlackFrob} one by one and recall their formulations.

\SlackFrobMainTh*

\begin{proof}
    Let $\BC$ be a base of $A$ with $\abs{\det A_{\BC}} = \Delta$.
    Note that
    \begin{equation*}
        \Delta_i(A A_{\BC}^{-1}) \leq 1,\quad \forall i \in \intint n.
    \end{equation*}
    By \eqref{eq:DiscDetBound} and \eqref{eq:DiscDetBoundReduced}, we have
  %   \begin{equation*}
  %       \gamma = C \cdot \begin{cases}
  %           \log k, \quad \text{for $k \leq n$,}\\
  %           \sqrt{\log k \cdot \log n},\quad \text{for $k \geq n$,}
		% \end{cases}
  %   \end{equation*}
  % \begin{equation*}
  %       \gamma = \JRConst \cdot \min\left\{
  %           \log k,
  %           \sqrt{\log k \cdot \log n}\right\}.
  %   \end{equation*}
    \begin{equation*}
        \herdisc(A_{\NotBC} A_{\BC}^{-1}) = O\left( \min\left\{
            \log k,
            \sqrt{\log k \cdot \log n} \right\} \right).
    \end{equation*}
    Now, the proof follows from \Cref{lm:SlackFrob_via_Discr}.
\end{proof}

\PolynomialSlackFrobMainTh*
\begin{proof}
    First, let us prove the Theorem with respect to $t_1$. By \Cref{lm:poly_subdet_search}, there exists a polynomial-time algorithm, which can construct a base $\BC$ of $A$, such that $\Delta_1(A A_{\BC}^{-1}) \leq e$. By \eqref{eq:SixDeviations} and \eqref{eq:SixDeviationsReduced}, we have
    \begin{gather*}
        \herdisc(A_{\NotBC} A_{\BC}^{-1}) = \begin{cases}
            O\bigl(\sqrt{k}\bigr),\quad\text{for $k \leq n$,}\\
            O\left(\sqrt{k} \cdot \log \bigl(\frac{2k}{n}\bigr)\right),\quad\text{for $k \geq n$.}
        \end{cases}
    \end{gather*}
    Applying \Cref{lm:SlackFrob_via_Discr}, we finish the proof for $t_1$.

    Next, we prove the theorem with respect to $t_2$ and $t_3$.
    Let $\BC$ be an arbitrary base of $A$, which can be found by a polynomial-time algorithm. 
    Note that
    \begin{equation*}
        \Delta_i(A A_{\BC}^{-1}) \leq \Delta/\delta_{\BC},\quad \forall i \in \intint n,
    \end{equation*}
    where $\delta_{\BC} = \abs{\det A_{\BC}}$.
    
    Therefore, by \eqref{eq:DiscDetBound} and \eqref{eq:DiscDetBoundReduced}, we can assume that
    \begin{equation*}
        \herdisc(A_{\NotBC} A_{\BC}^{-1}) = O\left( \min\left\{
            \log k,
            \sqrt{\log k \cdot \log n}\right\} \cdot \frac{\Delta}{\delta_{\BC}} \right).
    \end{equation*}
  % \begin{equation*}
  %       \gamma = \JRConst \cdot \min\left\{
  %           \log k,
  %           \sqrt{\log k \cdot \log n}\right\} \cdot \frac{\Delta}{\delta_{\BC}}.
  %   \end{equation*}
    Note that we can assume that $\delta_{\BC} \geq 2$, since, in the opposite case, $A$ is unimodular and the existence of an integer feasible solution is trivial. 
    Now, the proof again follows from \Cref{lm:SlackFrob_via_Discr}.
\end{proof}

\ExpSlackFrobMainTh*

\begin{proof}
    By \Cref{lm:exp_subdet_search_dual}, we can find a base $\BC$ of $A$, such that 
    \begin{equation*}
        \Delta_{i}\bigl(A A_{\BC}^{-1}\bigr) \leq e^{i+1},\quad\forall i \in \intint{n}
    \end{equation*}
    in $2^k \cdot \poly(\phi)$-time. By \eqref{eq:DiscDetBound} and \eqref{eq:DiscDetBoundReduced}, we have 
    \begin{equation*}
        \herdisc\bigl(A_{\NotBC} A_{\BC}^{-1}\bigr) = O\left( \min\left\{\log k, \sqrt{\log k \cdot \log n}\right\} \right).
    \end{equation*}
    % \begin{equation*}
    %     \herdisc\bigl(A A_{\BC}^{-1}\bigr) \leq e^2 \JRConst \cdot \min\left\{\log k, \sqrt{\log k \cdot \log n}\right\}.
    % \end{equation*}
    Now, the proof follows from \Cref{lm:SlackFrob_via_Discr}.
\end{proof}

Finally, we construct a lower bound that proves \Cref{prop:LBSlackFrob}.

\LBSlackFrob*

\begin{proof}
    Consider the polyhedron $\PC$ from \Cref{ex:gomory_is_tight}. It is easy to see that there exists a point $\hat x \in \PC$ whose slack with respect to each facet is exactly $(p-2)/2$. Indeed, choose $\hat x \in \RR^n$ such that 
    \begin{equation*}
        B \hat x + \frac{p-2}{2} \cdot \BUnit = b,
    \end{equation*}
    which exists as the unique solution to this system. Consequently, we get $c^\top \hat x - 1 = (p-2)/2$, confirming that $\hat x \in \PC$. Since $\PC \cap \ZZ^n = \emptyset$ and
    \begin{equation*}
        \binom{b}{-1} - \binom{B}{-c^\top} \hat x = \frac{p-2}{2} \cdot \BUnit,
    \end{equation*}
    we conclude that \(\SlackFrob(A) \geq (p-2)/2\) for \(A = \binom{B}{-c^\top}\).
\end{proof}

\subsection{Proofs with Respect to Systems in Standard Form}\label{sec:standard_proofs}

Here, we present the proof of \Cref{th:DiagFrob} and recall its definition. It directly uses \Cref{th:SlackFrob} and the polynomial-time reduction from \eqref{eq:Sys-SF} to \eqref{eq:Sys-CF}, provided in \Cref{lm:ILPSF_to_ILPCF} and \Cref{rm:StandardFormReduction}. The proofs of \Cref{th:PolyDiagFrob} and \Cref{th:ExpDiagFrob} can be deduced from \Cref{th:PolySlackFrob} and \Cref{th:ExpSlackFrob} in a similar way. Due to this reason, we skip them. 

\DiagFrobMainTh*

\begin{proof}
    By \Cref{lm:ILPSF_to_ILPCF} and \Cref{rm:StandardFormReduction}, the system \eqref{eq:Sys-SF} can be transformed to an equivalent system in the canonical form:
    \begin{equation}\label{eq:reduced_Sys-CF}
        \begin{cases}
            \hat A x \leq \hat b,\\
            x \in \ZZ^{\hat n},
        \end{cases}
    \end{equation}
    where $\hat n = n-k$, $\hat A \in \ZZ^{n \times \hat n}$, $\rank(\hat A) = \hat n$, and $\hat b \in \ZZ^n$. The new system satisfies the following properties:
    \begin{enumerate}
        \item $\Delta(A) = \Delta(\hat A)$,
        \item $\Delta_{\gcd}(\hat A) = 1$,
        \item each feasible base $\BC$ of \eqref{eq:Sys-SF} bijectively corresponds to a feasible base $\NotBC$ of \eqref{eq:reduced_Sys-CF},
        \item each feasible solution $y \in \RR_{\geq 0}^n$ of \eqref{eq:Sys-SF} bijectively corresponds to a feasible solution $x \in \RR^{\hat n}$ of \eqref{eq:reduced_Sys-CF} by the formula 
        \begin{equation}\label{eq:y_x_relation_DiagFrobProof}
            y = \hat b - \hat A x.
        \end{equation}
    \end{enumerate}
    Therefore, the condition
    \begin{equation}\label{eq:SlachCond_DiagFrobProof}
        \exists x \in \RR^{\hat n}:\quad \hat b - \hat A x \geq t \cdot \BUnit,
    \end{equation}
    is equivalent to the condition
    \begin{equation}\label{eq:DiagCond_DiagFrobProof}
        \exists y \in \RR_{\geq 0}^n: \quad b = A y,\quad y \geq t \cdot \BUnit.
    \end{equation}
    By \Cref{th:SlackFrob}, the condition \eqref{eq:SlachCond_DiagFrobProof} implies the existence of an integer feasible solution $\hat z \in \ZZ^{\hat n}$ of \eqref{eq:reduced_Sys-CF}. Moreover, the solution $\hat z$ can be found by a polynomial-time algorithm in the assumption that we know a base $\JC$ of $\hat A$, such that $\abs{\det \hat A_{\JC}} = \Delta$. By the described properties of $\hat A$, we can set $\JC := \NotBC$. Therefore, by \eqref{eq:y_x_relation_DiagFrobProof}, the condition \eqref{eq:DiagCond_DiagFrobProof} (which is equivalent to \eqref{eq:SlachCond_DiagFrobProof}) implies that there exists a feasible integer solution $z \in \ZZ_{\geq 0}^n$ of \eqref{eq:Sys-SF}, given by $z = \hat b - \hat A \hat z$, which can be constructed by a polynomial-time algorithm.
    
    % Therefore, the condition $A_{\BC}^{-1} b \geq (\Delta-1)\cdot \BUnit$ is equivalent to the condition $\hat b_{\NotBC} - \hat A_{\NotBC} x \geq (\Delta-1) \cdot \BUnit$ with respect to the system \eqref{eq:reduced_Sys-CF}.
\end{proof}

% \section{Cross referencing}\label{sec8}

% Environments such as figure, table, equation and align can have a label
% declared via the \verb+\label{#label}+ command. For figures and table
% environments use the \verb+\label{}+ command inside or just
% below the \verb+\caption{}+ command. You can then use the
% \verb+\ref{#label}+ command to cross-reference them. As an example, consider
% the label declared for Figure~\ref{fig1} which is
% \verb+\label{fig1}+. To cross-reference it, use the command 
% \verb+Figure \ref{fig1}+, for which it comes up as
% ``Figure~\ref{fig1}''. 

% To reference line numbers in an algorithm, consider the label declared for the line number 2 of Algorithm~\ref{algo1} is \verb+\label{algln2}+. To cross-reference it, use the command \verb+\ref{algln2}+ for which it comes up as line~\ref{algln2} of Algorithm~\ref{algo1}.

% \subsection{Details on reference citations}\label{subsec7}

% Standard \LaTeX\ permits only numerical citations. To support both numerical and author-year citations this template uses \verb+natbib+ \LaTeX\ package. For style guidance please refer to the template user manual.

% Here is an example for \verb+\cite{...}+: \cite{bib1}. Another example for \verb+\citep{...}+: \citep{bib2}. For author-year citation mode, \verb+\cite{...}+ prints Jones et al. (1990) and \verb+\citep{...}+ prints (Jones et al., 1990).

% All cited bib entries are printed at the end of this article: \cite{bib3}, \cite{bib4}, \cite{bib5}, \cite{bib6}, \cite{bib7}, \cite{bib8}, \cite{bib9}, \cite{bib10}, \cite{bib11}, \cite{bib12} and \cite{bib13}.

\section{Conclusion}\label{sec:conclusion}

In this work, we have derived new, significantly improved upper bounds on the diagonal Frobenius number of a matrix $A$, denoted by $\DiagFrob(A)$. Additionally, we have provided weaker bounds that admit polynomial-time algorithms for searching an integer feasible solution for systems in the standard form or slightly weaker bounds that admit $2^k \cdot \poly(\inputsize)$-time algorithms.

Furthermore, considering systems in the canonical form, we introduced a more general and natural diagonal Frobenius number for slacks, denoted by $\SlackFrob(A)$. In fact, all results were obtained for this generalized concept, and the corresponding results for the original diagonal Frobenius number $\DiagFrob(A)$ are merely corollaries of these results.

Note that we did not address the problem of constructing lower bounds for these numbers. This fact, along with further improvements of the upper bounds, may serve as a direction for future work.

\bibliographystyle{plainnat}
\bibliography{grib_biblio}

\begin{thebibliography}{34}
\providecommand{\natexlab}[1]{#1}
\providecommand{\url}[1]{\texttt{#1}}
\expandafter\ifx\csname urlstyle\endcsname\relax
  \providecommand{\doi}[1]{doi: #1}\else
  \providecommand{\doi}{doi: \begingroup \urlstyle{rm}\Url}\fi

\bibitem[Aggarwal et~al.(2024)Aggarwal, Joux, Santha, and W{\k{e}}grzycki]{ILPInTotalRegime}
Divesh Aggarwal, Antoine Joux, Miklos Santha, and Karol W{\k{e}}grzycki.
\newblock Polynomial time algorithms for integer programming and unbounded subset sum in the total regime.
\newblock \emph{arXiv preprint arXiv:2407.05435}, 2024.

\bibitem[Alekseev and Zakharova(2011)]{AZ}
E.~Alekseev, V. and V.~Zakharova, D.
\newblock Independent sets in the graphs with bounded minors of the extended incidence matrix.
\newblock \emph{Journal of Applied and Industrial Mathematics}, 5\penalty0 (1):\penalty0 14--18, 2011.
\newblock \doi{10.1134/S1990478911010029}.

\bibitem[Aliev and Henk(2010)]{DiagonalFrobenius}
Iskander Aliev and Martin Henk.
\newblock Feasibility of integer knapsacks.
\newblock \emph{SIAM Journal on Optimization}, 20\penalty0 (6):\penalty0 2978--2993, 2010.

\bibitem[Alon and Spencer(2016)]{AlonSpencerBook}
Noga Alon and Joel~H Spencer.
\newblock \emph{The probabilistic method}.
\newblock John Wiley \& Sons, 2016.

\bibitem[Artmann et~al.(2017)Artmann, Weismantel, and Zenklusen]{BimodularStrong}
Stephan Artmann, Robert Weismantel, and Rico Zenklusen.
\newblock A strongly polynomial algorithm for bimodular integer linear programming.
\newblock In \emph{Proceedings of the 49th Annual ACM SIGACT Symposium on Theory of Computing}, STOC 2017, pages 1206--1219, New York, NY, USA, 2017. Association for Computing Machinery.
\newblock ISBN 9781450345286.
\newblock \doi{10.1145/3055399.3055473}.

\bibitem[Bach et~al.(2025)Bach, Eisenbrand, Rothvoss, and Weismantel]{ForallPseudopoly}
Eleonore Bach, Friedrich Eisenbrand, Thomas Rothvoss, and Robert Weismantel.
\newblock Forall-exist statements in pseudopolynomial time.
\newblock In \emph{Proceedings of the 2025 Annual ACM-SIAM Symposium on Discrete Algorithms (SODA)}, pages 2225--2233. SIAM, 2025.

\bibitem[Birmpilis et~al.(2023)Birmpilis, Labahn, and Storjohann]{FastPSQDecomp}
Stavros Birmpilis, George Labahn, and Arne Storjohann.
\newblock A fast algorithm for computing the smith normal form with multipliers for a nonsingular integer matrix.
\newblock \emph{Journal of Symbolic Computation}, 116:\penalty0 146--182, 2023.

\bibitem[Bock et~al.(2014)Bock, Faenza, Moldenhauer, and Ruiz-Vargas]{StableSetHardness}
Adrian Bock, Yuri Faenza, Carsten Moldenhauer, and Andres~Jacinto Ruiz-Vargas.
\newblock {Solving the Stable Set Problem in Terms of the Odd Cycle Packing Number}.
\newblock In Venkatesh Raman and S.~P. Suresh, editors, \emph{34th International Conference on Foundation of Software Technology and Theoretical Computer Science (FSTTCS 2014)}, volume~29 of \emph{Leibniz International Proceedings in Informatics (LIPIcs)}, pages 187--198, Dagstuhl, Germany, 2014. Schloss Dagstuhl--Leibniz-Zentrum fuer Informatik.
\newblock ISBN 978-3-939897-77-4.
\newblock \doi{10.4230/LIPIcs.FSTTCS.2014.187}.

\bibitem[Dakhno et~al.(2024)Dakhno, Gribanov, Kasianov, Kats, Kupavskii, and Kuz'min]{HyperAvoiding_NonHomo}
Grigorii Dakhno, Dmitry Gribanov, Nikita Kasianov, Anastasiia Kats, Andrey Kupavskii, and Nikita Kuz'min.
\newblock Hyperplanes avoiding problem and integer points counting in polyhedra.
\newblock \emph{arXiv preprint arXiv:2411.07030}, 2024.

\bibitem[Fiorini et~al.(2022)Fiorini, Joret, Weltge, and Yuditsky]{TwoNonZerosStrong}
Samuel Fiorini, Gwena{\"e}l Joret, Stefan Weltge, and Yelena Yuditsky.
\newblock Integer programs with bounded subdeterminants and two nonzeros per row.
\newblock In \emph{2021 IEEE 62nd Annual Symposium on Foundations of Computer Science (FOCS)}, pages 13--24. IEEE, 2022.

\bibitem[Gomory(1965)]{GomoryRelation}
R.~E. Gomory.
\newblock On the relation between integer and noninteger solutions to linear programs.
\newblock \emph{Proceedings of the National Academy of Sciences}, 53\penalty0 (2):\penalty0 260--265, 1965.
\newblock ISSN 0027-8424.
\newblock \doi{10.1073/pnas.53.2.260}.
\newblock URL \url{https://www.pnas.org/content/53/2/260}.

\bibitem[Gomory(1969)]{GomoryCombinatorialPoly}
Ralph~E Gomory.
\newblock Some polyhedra related to combinatorial problems.
\newblock \emph{Linear algebra and its applications}, 2\penalty0 (4):\penalty0 451--558, 1969.

\bibitem[Gomory(1967)]{GomoryIntegerFaces}
RE~Gomory.
\newblock Integer faces of a polyhedron.
\newblock \emph{Proc. Natl. Acad. Sci. USA}, 57\penalty0 (1):\penalty0 16--18, 1967.

\bibitem[Gribanov et~al.(2024{\natexlab{a}})Gribanov, Malyshev, and Pardalos]{Gribanov_fixed_codim}
D~Gribanov, D~Malyshev, and Panos~M Pardalos.
\newblock Delta-modular {ILP} problems of bounded co-dimension, discrepancy, and convolution.
\newblock \emph{arXiv preprint arXiv:2405.17001}, 2024{\natexlab{a}}.

\bibitem[Gribanov(2023)]{SimplexEquiv_Gribanov}
Dmitry Gribanov.
\newblock Enumeration and unimodular equivalence of empty delta-modular simplices.
\newblock In \emph{International Conference on Mathematical Optimization Theory and Operations Research}, pages 115--132. Springer, 2023.

\bibitem[Gribanov et~al.(2024{\natexlab{b}})Gribanov, Shumilov, Malyshev, and Zolotykh]{SparseILP_Gribanov}
Dmitry Gribanov, Ivan Shumilov, Dmitry Malyshev, and Nikolai Zolotykh.
\newblock Faster algorithms for sparse {ILP} and hypergraph multi-packing/multi-cover problems.
\newblock \emph{Journal of Global Optimization}, pages 1--35, 2024{\natexlab{b}}.

\bibitem[Gribanov et~al.(2024{\natexlab{c}})Gribanov, Malyshev, Pardalos, and Zolotykh]{Parametric_Counting_Grib}
Dmitry~V Gribanov, Dmitry~S Malyshev, Panos~M Pardalos, and Nikolai~Yu Zolotykh.
\newblock A new and faster representation for counting integer points in parametric polyhedra.
\newblock \emph{Computational Optimization and Applications}, pages 1--51, 2024{\natexlab{c}}.

\bibitem[Gribanov and Veselov(2016)]{Width_Grib}
V.~Gribanov, D. and I.~Veselov, S.
\newblock On integer programming with bounded determinants.
\newblock \emph{Optim. Lett.}, 10:\penalty0 1169--1177, 2016.
\newblock \doi{10.1007/s11590-015-0943-y}.
\newblock URL \url{https://doi.org/10.1007/s11590-015-0943-y}.

\bibitem[Gribanov et~al.(2016)Gribanov, Malyshev, and Veselov]{WidthConv_Grib}
V.~Gribanov, D., S.~Malyshev, D., and I.~Veselov, S.
\newblock {FPT}-algorithm for computing the width of a simplex given by a convex hull.
\newblock \emph{Moscow University Computational Mathematics and Cybernetics}, 43\penalty0 (1):\penalty0 1--11, 2016.
\newblock \doi{10.3103/S0278641919010084}.

\bibitem[Gribanov et~al.(2022)Gribanov, Shumilov, Malyshev, and Pardalos]{OnCanonicalProblems_Grib}
V.~Gribanov, D., A.~Shumilov, I., S.~Malyshev, D., and M.~Pardalos, P.
\newblock On $\delta$-modular integer linear problems in the canonical form and equivalent problems.
\newblock \emph{J Glob Optim}, 2022.
\newblock \doi{10.1007/s10898-022-01165-9}.

\bibitem[Hu(1970)]{HuBook}
C.~Hu, T.
\newblock \emph{Integer programming and network flows}.
\newblock Addison-Wesley Publishing Company, London, 1970.

\bibitem[Jiang and Reis(2022)]{TightDiscDetBound}
Haotian Jiang and Victor Reis.
\newblock A tighter relation between hereditary discrepancy and determinant lower bound.
\newblock In \emph{Symposium on Simplicity in Algorithms (SOSA)}, pages 308--313. SIAM, 2022.

\bibitem[Lovász et~al.(1986)Lovász, Spencer, and Vesztergombi]{HerDisc}
L.~Lovász, J.~Spencer, and K.~Vesztergombi.
\newblock Discrepancy of set-systems and matrices.
\newblock \emph{European Journal of Combinatorics}, 7\penalty0 (2):\penalty0 151--160, 1986.
\newblock ISSN 0195-6698.
\newblock \doi{https://doi.org/10.1016/S0195-6698(86)80041-5}.
\newblock URL \url{https://www.sciencedirect.com/science/article/pii/S0195669886800415}.

\bibitem[Matou{\v{s}}ek(2013)]{DiscDetBound}
Ji{\v{r}}{\'\i} Matou{\v{s}}ek.
\newblock The determinant bound for discrepancy is almost tight.
\newblock \emph{Proceedings of the American Mathematical Society}, 141\penalty0 (2):\penalty0 451--460, 2013.

\bibitem[Nikolov(2015)]{LargestSimplex_Nikolov}
Aleksandar Nikolov.
\newblock Randomized rounding for the largest simplex problem.
\newblock In \emph{Proceedings of the forty-seventh annual ACM symposium on Theory of computing}, pages 861--870, 2015.

\bibitem[Nikolov(2018)]{DiscrepancyLections_Nikolov}
Aleksandar Nikolov.
\newblock Discrepancy theorya and applications.
\newblock \url{https://www.cs.toronto.edu/~anikolov/Prague18/}, 2018.
\newblock Accessed: August 24, 2025.

\bibitem[Oertel et~al.(2020)Oertel, Paat, and Weismantel]{DistributionsILP}
Timm Oertel, Joseph Paat, and Robert Weismantel.
\newblock The distributions of functions related to parametric integer optimization.
\newblock \emph{SIAM Journal on Applied Algebra and Geometry}, 4\penalty0 (3):\penalty0 422--440, 2020.
\newblock \doi{10.1137/19M1275954}.
\newblock URL \url{https://doi.org/10.1137/19M1275954}.

\bibitem[Paat et~al.(2021)Paat, Schl{\"o}ter, and Weismantel]{IntegralityNumber}
Joseph Paat, Miriam Schl{\"o}ter, and Robert Weismantel.
\newblock The integrality number of an integer program.
\newblock \emph{Mathematical Programming}, pages 1--21, 2021.
\newblock \doi{10.1007/s10107-021-01651-0}.
\newblock URL \url{https://doi.org/10.1007/s10107-021-01651-0}.

\bibitem[Shevchenko(1996)]{BlueBook}
Valery~N Shevchenko.
\newblock \emph{Qualitative topics in integer linear programming}.
\newblock American Mathematical Soc., Providence, Rhode Island, 1996.

\bibitem[Spencer(1985)]{SixDeviations_Spencer}
Joel Spencer.
\newblock Six standard deviations suffice.
\newblock \emph{Transactions of the American mathematical society}, 289\penalty0 (2):\penalty0 679--706, 1985.
\newblock \doi{10.1090/S0002-9947-1985-0784009-0}.
\newblock URL \url{https://doi.org/10.1090/S0002-9947-1985-0784009-0}.

\bibitem[Storjohann(1996)]{SNFOptAlg}
Arne Storjohann.
\newblock Near optimal algorithms for computing {S}mith normal forms of integer matrices.
\newblock In \emph{Proceedings of the 1996 International Symposium on Symbolic and Algebraic Computation}, ISSAC '96, pages 267--274, New York, NY, USA, 1996. Association for Computing Machinery.
\newblock ISBN 0897917960.
\newblock \doi{10.1145/236869.237084}.

\bibitem[Storjohann(2000)]{StorjohannDoc}
Arne Storjohann.
\newblock \emph{Algorithms for matrix canonical forms}.
\newblock PhD thesis, ETH Zurich, 2000.

\bibitem[Storjohann and Labahn(1996)]{HNFOptAlg}
Arne Storjohann and George Labahn.
\newblock Asymptotically fast computation of {H}ermite normal forms of integer matrices.
\newblock In \emph{Proceedings of the 1996 International Symposium on Symbolic and Algebraic Computation}, ISSAC '96, pages 259--266, New York, NY, USA, 1996. Association for Computing Machinery.
\newblock ISBN 0897917960.
\newblock \doi{10.1145/236869.237083}.

\bibitem[Veselov and Chirkov(2009)]{BimodularVert}
S.I. Veselov and A.J. Chirkov.
\newblock Integer program with bimodular matrix.
\newblock \emph{Discrete Optimization}, 6\penalty0 (2):\penalty0 220--222, 2009.
\newblock ISSN 1572-5286.
\newblock \doi{https://doi.org/10.1016/j.disopt.2008.12.002}.
\newblock URL \url{https://www.sciencedirect.com/science/article/pii/S1572528608000881}.

\end{thebibliography}

\end{document}